\documentclass{article}

\usepackage[margin=1in]{geometry}
\usepackage{amsmath, amssymb}
\usepackage{amsthm, thmtools}
\usepackage{bbold}
\usepackage{hyperref}
\usepackage{cleveref}
\usepackage{enumitem}
\usepackage{graphicx}
\usepackage{tikz}
\usetikzlibrary{arrows.meta,positioning,decorations.pathreplacing,calc,shapes}
\usepackage{array}
\usepackage{multirow}
\usepackage{xcolor}
\usepackage{soul}
\usepackage[backend=biber,style=alphabetic]{biblatex}
\usepackage{comment}
\usepackage{subcaption}
\usepackage{xspace}
\newtheorem{definition}{Definition}[section]
\newtheorem{lemma}[definition]{Lemma}

\newtheorem{corollary}[definition]{Corollary}

\newtheorem{proposition}[definition]{Proposition}
\newtheorem*{proposition*}{Proposition}
\newtheorem*{corollary*}{Corollary}

\usepackage{algorithm,caption}
\usepackage[noend]{algpseudocode}
\title{Complexity of Auctions with Interdependence}

\author{
Patrick Loiseau\\
INRIA, FairPlay team\\
\texttt{patrick.loiseau@inria.fr}
\and
Simon Mauras\\
INRIA, FairPlay team\\
\texttt{simon.mauras@inria.fr}
\and
Minrui Xu,\\
ENSAE, FairPlay team\\
\texttt{minrui.xu@ensae.fr}
}

\newcommand{\sigs}{{{\mathbf{s}}}}
\newcommand{\osigs}{{{\mathbf{s}_{-i}}}} % other bidders' signals
\newcommand{\vals}{{{\mathbf{v}}}}
\newcommand{\costs}{{{\mathbf{c}}}}

\newcommand{\allocs}{{{\mathbf{x}}}}
\newcommand{\payments}{{{\mathbf{p}}}}
\newcommand{\bids}{{{\mathbf{b}}}}

\newcommand{\polytope}{{{\mathcal T(\sigmas)}}}
\newcommand{\polytopetilde}{{{\mathcal T(\tilde{\sigmas})}}}
\newcommand{\intpolytope}{{{\mathcal T_D(\sigmas)}}}
\newcommand{\polyver}{{{\mathcal V(\sigmas)}}}
\newcommand{\polyvertilde}{{{\mathcal V(\tilde{\sigmas})}}}
\newcommand{\avg}[2]{{\langle#1,#2\rangle}}
\newcommand{\cfpo}{(s_1,s_2)}
\newcommand{\cfpt}{(s_1^{\prime},s_2^{\prime})}

\newcommand{\sigmas}{\boldsymbol\sigma}
\newcommand{\rhos}{\boldsymbol\rho}
\newcommand{\paras}{(\boldsymbol\rho,\boldsymbol\sigma)}
\newcommand{\Rv}{R_V}
\newcommand{\Rvopt}{R_V^*(\vals)}

\newcommand{\Rc}{R_C}
\newcommand{\Rcopt}{R_C^*(\costs)}

\newcommand{\Rd}{R_D}
\newcommand{\Rdoptv}{R_D^*(\vals)}

\newcommand{\ind}[1]{\mathbb 1[{#1}]}

\newcommand{\PbD}{\textsc{Det}\xspace}
\newcommand{\PbC}{\textsc{Cst}\xspace}
\newcommand{\PbV}{\textsc{Val}\xspace}

\date{}  % No date shown

\begin{document}

\begin{titlepage}

\maketitle

\begin{abstract}
We study auction design in the celebrated interdependence model introduced by Milgrom and Weber [1982], where a mechanism designer allocates a good, maximizing the value of the agent who receives it, while inducing truthfulness using payments. In the lesser-studied procurement auctions, one allocates a chore, minimizing the cost incurred by the agent selected to perform it.

Most of the past literature in theoretical computer science considers designing truthful mechanisms with constant approximation for the value setting, with restricted domains and monotone valuation functions. 

In this work, we study the general computational problems of optimizing the approximation ratio of truthful mechanism, for both value and cost, in the deterministic and randomized settings. Unlike most previous works, we remove the domain restriction and the monotonicity assumption imposed on value functions. We provide theoretical explanations for why some previously considered special cases are tractable, reducing them to classical combinatorial problems, and providing efficient algorithms and characterizations. We complement our positive results with hardness results for the general case, providing query complexity lower bounds, and proving the NP-Hardness of the general case.
\end{abstract}

\end{titlepage}

%%%%%%%%%%%%%%%%%%%%%%%%%%%%%%%%%%%%%%%%%%%%%%%%%%%%%%%%%%%%%%%
%%%%%%%%%%%%%%%%%%%%%%%%%%%%%%%%%%%%%%%%%%%%%%%%%%%%%%%%%%%%%%%
%%%%%%%%%%%%%%%%%%%%%%%%%%%%%%%%%%%%%%%%%%%%%%%%%%%%%%%%%%%%%%%
\section{Introduction}\label{sec:intro}

%%%%%%%%%%%%%%%%%%%%%%%%%%%%%%%%%%%%%%%%%%%%%%%%%%%%%%%%%%%%%%%
\paragraph{Algorithmic mechanism design.} Algorithmic mechanism design lies at the intersection of computer science and economics, focusing on the development of algorithms that account for the strategic behaviors of self-interested agents \cite{NisanR99,NRTV2007}. A central problem in this area is to design auction mechanisms to select the most suitable bidder among $n$ agents. Notably, in the second-price auction (or Vickrey auction), where the winner pays the second largest value, it is optimal for bidders with \emph{independent values} to bid truthfully \cite{vickrey1961counterspeculation}. However, second-price auction fails to preserve truthfulness in more intricate scenarios where a bidder's valuation depends on information held by others, a situation commonly referred to as \emph{interdependent values} \cite{milgrom1982theory},  Formally, each agent $i\in[n]$ possesses a private signal $s_i\in[k]$, and a publicly known function $v_i: [k]^n \rightarrow \mathbb R_{>0}$ aggregating the signals into a value.

As an illustration, consider a situation where $n=2$ agents, Alice and Bob, are competing for a single good, say a small house in the countryside. Alice's value for the house highly depends on its condition, but this information is privately held by Bob; while Bob's value for the house is moderate, as he plans to rebuild it entirely. Let $k = 2$, and define the value functions as $v_1(s_1, s_2) = 1+99(s_2-1)$ and $v_2(s_1, s_2) = 10$, that is, Alice's value is $100$ if Bob's signal is high ($s_2 = 2$), and $1$ otherwise ($s_2 = 1$); while Bob's value is always $10$ regardless of the signals. If Bob reports truthfully a high signal ($s_2 = 2$), a second-price auction would allocate Alice the house, at a price of 10. However, notice that Bob could win the auction by misreporting a low signal ($s_2 = 1$). One truthful alternative is to flip a fair coin and allocate the house at random. While not optimal, this approach guarantees that the expected value of the winner is at least half of the optimal, yielding an approximation ratio\footnote{the approximation ratio of a randomized mechanism is the worst (over possible signals) ratio between the optimal value/cost and the expectation of the value/cost selected by the mechanism.} of $2$.

%%%%%%%%%%%%%%%%%%%%%%%%%%%%%%%%%%%%%%%%%%%%%%%%%%%%%%%%%%%%%%%
\paragraph{Auctions with interdependence.}
A recent line of work has studied the design of auctions with interdependence \cite{RoughgardenT16,EdenFFG18,EdenFFGK19,AmerT21,LuSZ22}, achieving constant factor approximation ratios on restricted domains of valuation functions. In particular, the optimal approximation ratio with submodular\footnote{see \cite{EdenFFGK19} for the definition of submodular over signals.} valuation functions is at least $2$ \cite{EdenFFGK19}, at most $3.315$ \cite{LuSZ22}, and equal to $2$ in the special cases where $n=2$ \cite{EdenFFG18} or $k=2$ \cite{AmerT21}.  Interestingly, this literature focuses on auctions for goods, but not on reverse auctions (also known as procurement auctions), which allocate chores, i.e., tasks or services to be performed, with payments made to cover the winning contractor's cost. Most of the underlying economic theory used to guarantee truthfulness, is identical in the standard and reverse auctions \cite{krishna2009auction}, when using cost functions $c_i: [k]^n \rightarrow \mathbb R_{>0}$.

Returning to our example, suppose Alice and Bob are now contractors bidding to renovate the countryside house, and assume their costs are given by $c_1(s_1, s_2) = 1 + 99(2-s_2)$ for Alice, who will renovate it from its current state; and $c_2(s_1, s_2) = 10$ for Bob who will rebuild it from scratch. As before, reverse second-price auction is vulnerable to manipulation, as Bob could report $s_2=1$ instead of $s_2 = 2$ to be awarded the contract with a payment of $100$. But unlike the value-setting, allocating uniformly at random is not acceptable here. Indeed, if $s_2 = 2$, the optimal cost is $c_1(s_1,2) = 1$, but the expected cost of a random allocation is $(c_1(s_1,2) + c_2(s_1,2))/2 = 5.5$, which can lead to an arbitrarily large approximation ratio when modifying the parameters.

%%%%%%%%%%%%%%%%%%%%%%%%%%%%%%%%%%%%%%%%%%%%%%%%%%%%%%%%%%%%%%%
\paragraph{Main contributions.}
\begin{itemize}
\item
First, observe from our toy examples that the chore setting is more challenging from an approximation perspective. Intuitively, most mechanisms \cite{EdenFFGK19} perform well with constant probability, which suffices to ensure constant approximations in the value setting, but not in the cost setting.
In this paper, \ul{we initiate the study of approximate procurement auctions with interdependence.} 
Second, notice that our toy examples use monotone value and cost functions, as signals represent a quantitative information about the good or chore, which is a standard assumption made in all the previous work cited above. However, in many practical scenarios the value or cost functions are not monotone, for example when the agents have opposite preferences, or when signals cannot be compared as they capture qualitative information. In this paper, \ul{we generalize the characterization of truthful mechanisms of \mbox{\cite{RoughgardenT16}} to non-monotone settings.}

\item
While previous work \cite{EdenFFG18,EdenFFGK19,AmerT21,LuSZ22} has focused on designing greedy constant-factor approximation mechanisms for allocating goods under restricted domains of value functions, \ul{we investigate the optimization problem of computing mechanisms that achieve the best possible (i.e., smallest) approximation ratio}. Specifically, we study this problem in three distinct settings: randomized mechanisms for goods, randomized mechanisms for chores, and deterministic mechanisms. We propose several frameworks, reducing our optimization problems from and to classic combinatorial problems, which provide a better understanding of the general domain, and could provide improved greedy mechanisms in the aforementioned restricted domains.

\end{itemize}

%%%%%%%%%%%%%%%%%%%%%%%%%%%%%%%%%%%%%%%%%%%%%%%%%%%%%%%%%%%%%%%
%%%%%%%%%%%%%%%%%%%%%%%%%%%%%%%%%%%%%%%%%%%%%%%%%%%%%%%%%%%%%%%
\subsection{Our result}

We define the optimization problems $\PbV$ and $\PbC$, whose objectives are to compute the randomized mechanism which achieves the smallest (i.e., optimal) approximation ratio, respectively in the value and cost settings. Additionally, we consider deterministic mechanisms, for which the approximations in value and cost are the same, and we denote the corresponding optimization problem as $\PbD$. To measure the efficiency of an algorithm solving our minimization problems, we measure the time complexity with respect to the total size $N = nk^nb$ of the input, where $b$ denotes the maximum number of bits used to represent a value or a cost.

\begin{restatable}{theorem}{thmtwoagents}\label{thm:twoagents}
When $n=2$, one can solve $\PbV$, $\PbC$ and $\PbD$ in $O(N \log N)$ time.
\end{restatable}

\begin{restatable}{theorem}{thmtwosignals}\label{thm:twosignals}
When $k=2$, one can solve $\PbD$ in $N^{1+o(1)}$ time.
\end{restatable}

\begin{restatable}{theorem}{thmlp}
\label{thm:lp}
One can solve $\PbV$ and $\PbC$ in $N^{O(1)}$ time.
\end{restatable}

Next, to prove hardness we define the gap variants of our optimization problem: given a minimization problem \textsc{Pb} and two constants $0\leq\alpha \leq \beta$, the gap variant $(\alpha,\beta)$-\textsc{Pb} is a decision problem which asks to distinguish between instances where the optimal solution has measure  $\leq\alpha$ and $>\beta$ (the algorithm can fail or not terminate on intermediate instances).

\begin{restatable}{theorem}{thmnphard}\label{thm:nphard}
When $n = 4$, the problem $(1,\beta)$-\PbD is NP-Hard for any $\beta > 1$.
\end{restatable}

However, in mechanism design, we are mostly interested in computing the outcome of the mechanism at the reported signal profile. Thus, for each minimization problem \textsc{Pb}, we introduce the query problem \textsc{Pb}$_\gamma$, which asks to answer queries about a solution of measure $\leq\gamma$, provided that one exists (the algorithm can fail or not terminate otherwise), while being consistent across queries. In our context, one query corresponds to computing the outcome at one signal profile. We establish the following corollary of \Cref{thm:lp,thm:twoagents,thm:twosignals}.

\begin{corollary*}
For every $\gamma > 1$, and measuring the complexity w.r.t. $N$, we have that:
\begin{itemize}
    \item when $n=2$ or $k=2$, one can answer $\PbD_\gamma$ queries in quasilinear time,
    \item one can answer $\PbV_\gamma$ and $\PbC_\gamma$ queries in polynomial time,
    \item assuming $\text{P}\neq\text{NP}$, one cannot always answer $\PbD_\gamma$ queries in polynomial time, even on instances where there exists a deterministic allocation rule of ratio $1$.
\end{itemize}

\end{corollary*}

Finally, one might wonder if the complexities we obtained are improvable, as answering one query might require less time than computing the entire allocation rule. However, we show that the exponential dependency in $n$ is unavoidable, by having the input being accessible through an oracle, and counting the number of oracle queries necessary to solve a computational task.

\begin{restatable}{theorem}{thmquery}\label{thm:query}
For all fixed $n\geq 2$ and $k\geq 2$, it requires at least $\Omega(k^n)$ queries to the input oracle (value or cost) to answer some $\PbV_\gamma$, $\PbC_\gamma$ and $\PbD_\gamma$ queries.
\end{restatable}

\begin{table}[h!]
    \centering
    \def\arraystretch{1.2}
    \begin{tabular}{|c|l|l|}
    \cline{2-3}
    \multicolumn{1}{c|}{} & \multicolumn{1}{c|}{Time complexity} & \multicolumn{1}{c|}{Query complexity} \\
    \hline
        $n = 2$ &\Cref{thm:twoagents}: \PbV, \PbC, \PbD in $O(N\log N)$ &  \\
    \cline{1-2}
        $k = 2$ &\Cref{thm:twosignals}: \PbD in $N^{1+o(1)}$ & \Cref{thm:query}: \PbV, \PbC, \PbD \\
    \cline{1-2}
        general & \Cref{thm:lp}: \PbV and \PbC in $N^{O(1)}$ & require  $\Omega(k^n)$ queries \\
        case & \Cref{thm:nphard}: \PbD is NP-Hard &   \\
    \hline
    \end{tabular}
    \caption{Summary of our main results. The time complexity is given as a function of the total input size $N$, and the query complexity is the number of value or cost oracle queries.}
    \label{tab:results}
\end{table}

%%%%%%%%%%%%%%%%%%%%%%%%%%%%%%%%%%%%%%%%%%%%%%%%%%%%%%%%%%%%%%%
%%%%%%%%%%%%%%%%%%%%%%%%%%%%%%%%%%%%%%%%%%%%%%%%%%%%%%%%%%%%%%%
\subsection{Related Works}

The interdependent value model stemmed from modeling settings such as mineral rights auctions, where bidders have a common value \cite{wilson1969communications, milgrom1979convergence,dasgupta2000efficient}, and was formalized by Milgrom and Weber \cite{milgrom1982theory}. This model was extensively studied by economists, who characterized the class of valuation functions for which allocating to the highest value can be implemented truthfully, using the ``single-crossing'' condition \cite{maskin1996auctions,ausubel1999generalized,ausubel1999generalized,jehiel2001efficient}.

Auctions with interdependent values received recent attention from the theoretical computer science community \cite{ChawlaFK14,RoughgardenT16}, considering the task of approximately maximizing revenue (payment made to the seller of a good) and social welfare (value of a buyer), beyond the single-crossing assumption. In their founding work, Roughgarden and Talgam-Cohen \cite{RoughgardenT16} characterized the allocation rules which can be implemented truthfully, when agents have increasing value functions, using a generalization of Myerson's lemma \cite{myerson1981optimal}. In \Cref{lem:truthful}, we generalize their characterization to non-monotone value and cost functions. The only prior work which manages to deal with non-monotone value functions is \cite{EdenGZ22}, but they only provide a sufficient condition, potentially missing out on some truthful allocation rules which do not satisfy it.

Closest to our work is a recent line of research that develops constant-approximation mechanisms for agents with interdependent values belonging to restricted classes, such as ``$c$-single-crossing'' \cite{EdenFFG18} or ``submodular-over-signals'' \cite{EdenFFGK19,AmerT21,LuSZ22}. In particular, Eden, Feldman, Fiat, and Goldner \cite{EdenFFG18} study the special cases of \PbD with either $n = 2$ agents or $k = 2$ signals under monotone value functions. They introduce the $c$-single-crossing condition, which is sufficient to guarantee the existence of a $c$-approximate deterministic allocation rule. In contrast, \Cref{thm:twoagents,thm:twosignals} consider the same special cases and construct deterministic allocation rules achieving the best possible approximation ratio, with comparable computational complexity, without restricting to the $c$-single-crossing domain and without assuming monotonicity. Moreover, in the case $n = 2$, we provide an exact characterization of the value functions that admit a $c$-approximation; as shown in \Cref{prop:single-crossing}, this characterization strictly generalizes the $c$-single-crossing condition (which remains sufficient as a special case). Finally, prior work on \PbV provide randomized allocation rules when valuation functions are monotone and submodular-over-signals, establishing that the optimal approximation ratio is at most $3.315$ in general \cite{LuSZ22}, at most $2$ when $k = 2$ \cite{AmerT21}, and exactly $2$ in the worst case when $n = k = 2$ \cite{EdenFFGK19}.

Follow-up works have studied simple auction formats such as clock-auctions \cite{GkatzelisPPS21,FeldmanGG022}, or more general scenarios such as online auctions \cite{MaurasMR24, FeldmanMMR25}, or with private valuation functions \cite{EdenGZ22,EdenFGMM23,EdenFMM24}. Finally, our paper bears resemblance with work studying the computational and communication complexity of truthful mechanism
\cite{Dobzinski11,DobzinskiV12,DobzinskiD13,DaskalakisDT14,Dobzinski16,AssadiKSW20,RubinsteinZ21,BabaioffDR23}.

%%%%%%%%%%%%%%%%%%%%%%%%%%%%%%%%%%%%%%%%%%%%%%%%%%%%%%%%%%%%%%%
%%%%%%%%%%%%%%%%%%%%%%%%%%%%%%%%%%%%%%%%%%%%%%%%%%%%%%%%%%%%%%%
\subsection{Organization}

The rest of the paper is organized as follows: \Cref{sec:model} introduces the model and defines important notations, \Cref{sec:ideas} presents an overview of the main ideas for our technical contributions, \Cref{sec:positive-results} gives the formal proofs of our positive results, and \Cref{sec:negative-results} gives the formal proofs of our negative results. 

%%%%%%%%%%%%%%%%%%%%%%%%%%%%%%%%%%%%%%%%%%%%%%%%%%%%%%%%%%%%%%%
%%%%%%%%%%%%%%%%%%%%%%%%%%%%%%%%%%%%%%%%%%%%%%%%%%%%%%%%%%%%%%%
%%%%%%%%%%%%%%%%%%%%%%%%%%%%%%%%%%%%%%%%%%%%%%%%%%%%%%%%%%%%%%%
\section{Model and Preliminaries} \label{sec:model}

We consider an auction setting where $n$ agents compete to be selected, either to receive a good (value maximization) or to perform a chore (cost minimization). Denoting $[\ell] := \{1, \dots, \ell\}$, each agent $i \in [n]$ has a (private) signal $s_i \in [k]$ 
which captures the information she has
on being selected. For convenience, we will denote $\sigs := (s_i)_{i\in [n]}$ and $\sigs_{-j} := (s_i)_{i\in [n]\setminus\{j\}}$. 

%%%%%%%%%%%%%%%%%%%%%%%%%%%%%%%%%%%%%%%%%%%%%%%%%%%%%%%%%%%%%%%
\paragraph{Allocation rule.} A \emph{randomized} allocation rule is a collection of functions $\allocs := (x_i)_{i\in [n]}$, where $x_i: [k]^n \rightarrow [0,1]$ gives the probability of agent $i$ being selected. We impose the constraint that for all realizations of signals the probabilities must sum up to one:
\begin{equation}
\forall \sigs \in [k]^n,\qquad \sum_{i\in[n]} x_i(\sigs) = 1.
\label{eq:proba}
\end{equation}
An allocation rule $\allocs$ is \emph{deterministic} if all $x_i$'s have value in $\{0,1\}$.

%%%%%%%%%%%%%%%%%%%%%%%%%%%%%%%%%%%%%%%%%%%%%%%%%%%%%%%%%%%%%%%
\paragraph{Value and cost.} Each agent has a (publicly known) function which aggregates the information of all agents into a single parameter:
\begin{itemize}
\item when allocating a \emph{good}, each agent $i\in [n]$ has a value function $v_i: [k]^n \rightarrow \mathbb R_{>0}$;
\item when allocating a \emph{chore}, each agent $i\in [n]$ has a cost function $c_i: [k]^n \rightarrow \mathbb R_{>0}$.
\end{itemize}

Most of the previous works \cite{EdenFFG18,EdenFFGK19,AmerT21} focus on non-decreasing value functions, that is, $\forall i,j \in [n],\sigs_{-j}\in[k]^{n-1}$, and $\forall s_j\in [k-1]$, $v_i(s_j,\sigs_{-j})\le v_i(s_j+1,\sigs_{-j})$. Similarly, we consider non-increasing cost functions, defined by $c_i(s_j,\sigs_{-j})\ge c_i(s_j+1,\sigs_{-j})$.

%%%%%%%%%%%%%%%%%%%%%%%%%%%%%%%%%%%%%%%%%%%%%%%%%%%%%%%%%%%%%%%
\paragraph{Performance ratios and signal orderings.} To unify the good and chore settings, we will assume (without loss of generality) that we are given value and cost functions which satisfy: $$\forall i\in[n],\forall \sigs\in[k]^n,\qquad v_i(\sigs) \cdot c_i(\sigs) = 1.$$
Abusing notation, we might use $\vals$ in the chore setting, in which case the cost should be computed as the inverse of the value. Then, to abstract away the specific setting, we compute two quantities: performance ratios and signal orderings. The performance ratios are normalized values and costs, that will be used to compute the performance of our allocation:
\begin{equation}
\forall i\in [n],\forall \sigs \in [k]^n,\qquad
\rho_i(\sigs) :=
\frac{v_i(\sigs)}{\max_{j\in[n]} v_j(\sigs)} =
\frac{\min_{j\in[n]} c_j(\sigs)}{c_i(\sigs)}\in (0,1].
\tag{$\boldsymbol{\rho}$}\label{eq:ratios}
\end{equation}
We switch our focus from value functions (resp. cost functions) to performance ratios in the following, and claim by the following lemma that this is exactly a reformulation of value (resp. cost) setting.

\begin{restatable}{lemma}{lemvaluationgeneral}
\label{lm_find_value_for_performance_ratios}
Given a collection of ratios $\rhos=(\rho_i(\sigs))_{i\in[n], \sigs\in [k]^n}$, where for all signal profiles $\sigs$ and bidders $i$, $\rho_i(\sigs)\in (0,1]$, and for each signal profile $\sigs$, there exists at least one bidder $i$ with $\rho_i(\sigs)=1$, there exists increasing value (resp. decreasing cost) functions which induce the given ratios.
\end{restatable}

\begin{proof}
    %It suffices to provide construction only for the value setting. Indeed, given any increasing value functions inducing the ratios, taking its reciprocal yields a construction for cost setting.
    We denote the minimal value of $\rhos$ by ${r}^*(\rhos)$.
    It is easy to check that the following constructions satisfy all the requirements:
    \begin{itemize}
        \item for value setting, $v_i(\sigs)=\rho_i(\sigs)/({r}^*(\rhos)/2)^{\ell}$, where $\ell = ||\sigs||_1$;
        \item for cost setting, $c_i(\sigs)=({r}^*(\rhos)/2)^{\ell}/\rho_i(\sigs)$, where $\ell = ||\sigs||_1$.
    \end{itemize}
\end{proof}
%The proof is deferred to \Cref{ssec:build_valuation}.

Next, the signal ordering will allow us to characterize truthfulness, without the assumption that valuation functions are monotone, as was done by most previous works \cite{RoughgardenT16, EdenFFG18, EdenFFGK19, AmerT21}. 
For each agent $i\in[n]$ and signals $\osigs\in[k]^{n-1}$, we construct a binary relation $\sigma_i(\osigs)$ over $s_i \in [k]$ by sorting in non-decreasing (partial) order of value, or in non-increasing (partial) order of cost.

\begin{align*}
\forall i\in[n],\forall \osigs\in[k]^{n-1},\qquad
\sigma_i(\osigs)
:=& \;\{(s_i, s_i')\in[k]^2\;|\; v_i(s_i, \osigs) < v_i(s_i',\osigs)\}
\tag{$\boldsymbol{\sigma}$}\label{eq:orders}\\
=& \;\{(s_i, s_i')\in[k]^2\;|\; c_i(s_i, \osigs) > c_i(s_i',\osigs)\}.
\end{align*}

We observe that for all $i$ and $\sigs_{-i}$ the binary relation $\sigma_i(\sigs_{-i})$ is a {strict} order: for all $s_i$ and $s_i'$, either $(s_i,s_i')\notin\sigma_i(\osigs)$ or $(s_i',s_i)\notin\sigma_i(\osigs)$. However, this order may not be total: there might exist $s_i$ and $s_i'$ such that $(s_i, s_i')\notin \sigma_i(\osigs)$ and $(s_i',s_i)\notin\sigma_i(\osigs)$, which happens when $v_i(s_i,\sigs_{-i})=v_i(s_i^{\prime},\sigs_{-i})$ or $c_i(s_i,\sigs_{-i})=c_i(s_i^{\prime},\sigs_{-i})$.

%%%%%%%%%%%%%%%%%%%%%%%%%%%%%%%%%%%%%%%%%%%%%%%%%%%%%%%%%%%%%%%
%%%%%%%%%%%%%%%%%%%%%%%%%%%%%%%%%%%%%%%%%%%%%%%%%%%%%%%%%%%%%%%
\subsection{Truthfulness}

Each agent $i\in [n]$ will report a bid $b_i\in[k]$, which may not be equal to their private signal $s_i$. To incentivize agents to report their true signal, we design mechanisms, which are allocation rules endowed with payment functions $\payments := (p_i)_{i\in[n]}$ with $p_i: [k]^n \rightarrow \mathbb R_{\geq 0}$, which are either charged or transferred to the agents depending on the scenario (good or chore). We assume that agents are rational, and act to maximize the \emph{quasi-linear utility}:
\begin{itemize}
\item when allocating a \emph{good}, each agent $i\in[n]$ has utility $u_i(\bids; \sigs) := x_i(\bids) \cdot v_i(\sigs) - p_i(\bids)$;
\item when allocating a \emph{chore}, each agent $i\in[n]$ has utility $u_i(\bids; \sigs) := p_i(\bids) - x_i(\bids) \cdot c_i(\sigs)$.
\end{itemize}
A mechanism is \emph{truthful} (also known as EPIC, for \emph{ex-post incentive compatible}) if reporting the true signal is a Nash-equilibrium, that is, if for each agent $i$ reporting $b_i = s_i$ is a best response when all other agents report $\bids_{-i} = \sigs_{-i}$. More formally:
\begin{align}
\forall \sigs\in[k]^n, \forall i\in [n],\forall b_i\in[k], \qquad u_i(s_i, \sigs_{-i}; \sigs) &\geq u_i(b_i, \sigs_{-i}; \sigs) \tag{IC}\label{eq:IC}\\
\forall \sigs\in[k]^n, \forall i\in [n], \qquad u_i(s_i, \sigs_{-i}; \sigs) &\geq 0 \tag{IR}\label{eq:IR}
\end{align}

Previous work \cite{RoughgardenT16} has characterized allocation rule that can be made truthful in expectation (when allocating a good) when the value function is non-decreasing in the signals. We refine the characterization so that it applies to a more general setting allowing for non-monotone value or cost functions.  

\begin{restatable}[adapted from \cite{RoughgardenT16}]{lemma}{lemtruthful}\label{lem:truthful}
    An allocation rule $\allocs$ can be implemented truthfully if and only if it satisfies the following monotonicity property:
    \begin{equation}
    \forall i\in[n], \forall\osigs\in[k]^{n-1},\forall (s_i, s_i')\in \sigma_i(\osigs),\qquad x_i(s_i, \sigs_{-i}) \leq x_i(s_i', \sigs_{-i}).
    \label{eq:monotonicity}
    \end{equation}
\end{restatable}

\begin{proof}
Given such an allocation rule $\allocs$, for each agent $i\in[n]$ and $\osigs\in[k]^{n-1}$ we extend the partial order $\sigma_i(\osigs)$ into a total order $\tau_i(\osigs)$ monotone with respect to $x_i(\cdot, \osigs)$, that is, such that for all $(s_i, s_i')\in \tau_i(\osigs)$ we have $x_i(s_i, \osigs) \leq x_i(s_i', \osigs)$. From the definition of $\sigma_i(\osigs)$ observe that for all $(s_i, s_i')\in\tau_i(\osigs)$ we also have \begin{align*}
v_i(s_i, \osigs) &\leq v_i(s_i', \osigs)
&\text{(good),}\\
c_i(s_i, \osigs) &\geq c_i(s_i', \osigs)
&\text{(chore).}
\end{align*}
For convenience, given $\sigs\in[k]^n$ we define the set of signals which are ranked before $s_i$ in $\tau_i(\osigs)$
$$
\forall i\in[n],\forall \sigs\in[k]^n,\qquad T_i(\sigs) := \{s_i\}\cup\{t\in[k]\;|\; (t,s_i)\in \tau_i(\osigs) \}
$$
Next, we define the increase in allocation probability:
$$
\forall i\in[n],\forall \sigs\in[k]^n,\quad \delta x_i(\sigs):=
x_i(\sigs) - \max(\{0\}\cup\{x_i(t, \osigs)\;|\; t\in T_i(\sigs)\setminus\{s_i\}\})
$$
In particular, observe that we have
$$
\forall i\in[n],\forall \sigs\in[k]^n,\qquad
x_i(\sigs) = \sum_{t\in T_i(\sigs)} \delta x_i(t, \osigs)
$$
Finally, we define the payment functions:
\begin{align*}
p_i(\bids) :=& \sum_{t\in T_i(\bids)} \delta x_i(t,\bids_{-i})\cdot v_i(t, \bids_{-i}) & \text{(good)},\\
p_i(\bids) :=& \sum_{t\in T_i(\bids)} \delta x_i(t,\bids_{-i})\cdot c_i(t, \bids_{-i})& \text{(chore)}.
\end{align*}

To check that properties (\ref{eq:IC}) and (\ref{eq:IR}) are verified, we compute the utility $u_i(b_i, \osigs; \sigs)$.
\begin{align*}
u_i(b_i, \osigs; \sigs) &= \sum_{t\in T_i(b_i, \osigs)} \delta x_i(t,\sigs_{-i}) \cdot ({v_i(\sigs) - v_i(t, \osigs)}) & (\text{good}),\\
u_i(b_i, \osigs; \sigs) &= \sum_{t\in T_i(b_i, \osigs)} \delta x_i(t,\sigs_{-i}) \cdot ({c_i(t, \osigs) - c_i(\sigs)}) & \text{(chore)}.
\end{align*}
Observe that the bid $b_i$ only affects the set $T_i(b_i, \osigs)$ on which we compute the sum, but the summand does not depend in $b_i$. Moreover, each non-zero term $t\in T_i(b_i, \osigs)$ is positive if and only if $t\in T_i(s_i, \osigs)$, thus $u_i(b_i, \osigs; \sigs)$ is non-negative and maximized when $b_i = s_i$, proving (\ref{eq:IR}) and (\ref{eq:IC}).

Now we prove the only if direction. Assume that property (\ref{eq:IC}) holds. By definition, for all $s_i,s_i'\in [k]$, we have the following inequalities
\begin{align*}
u_i(s_i, \osigs; s_i, \osigs) &\geq u_i(s_i', \osigs; s_i, \osigs)\\
u_i(s_i', \osigs; s_i', \osigs) &\geq u_i(s_i, \osigs; s_i', \osigs)
\end{align*}
Summing the inequalities and rearranging the terms, we get
\begin{align*}
(x_i(s_i, \sigs_{-i})-x_i(s_i^{\prime}, \sigs_{-i})) \cdot \bigl(v_i(s_i, \sigs_{-i}) - v_i(s_i^{\prime}, \sigs_{-i})\bigr) 
&\ge 0 &&\text{(good),}\\
(x_i(s_i, \sigs_{-i})-x_i(s_i^{\prime}, \sigs_{-i})) \cdot \bigl(c_i(s_i, \sigs_{-i}) - c_i(s_i^{\prime}, \sigs_{-i})\bigr) 
&\le 0 &&\text{(chore).}
\end{align*}
which implies the monotonicity property.
\end{proof}

Recall that the (partial) orders over signals $\sigmas$ have been constructed from the value or cost function. In the rest of the paper, we will not discuss payment functions and agents' utilities, instead,  we will focus on finding allocation rules which belong to the following polytopes.
\begin{definition}
    Define the truthful polytope $\polytope$ as the set of all monotone allocation rules:
    \begin{equation*}
        \polytope := \left\{(x_i(\sigs))_{i\in[n]}^{\sigs\in[k]^n} \;\middle|\; \begin{array}{ll}
        x_i(\sigs) \geq 0&\forall i\in[n],\forall \sigs\in[k]^n\\
        \sum_{i\in[n]} x_i(\sigs) = 1&\forall \sigs\in[k]^n\\
        x_i(s_i,\osigs) \leq x_i(s_i',\osigs)&\forall i\in[n],\forall \osigs\in[k]^{n-1},\forall (s_i, s_i')\in\sigma_i(\osigs)\\
        \end{array}\right\}.
    \end{equation*}
We further define the set of all deterministic monotone allocations:
\begin{equation*}
    \intpolytope := \left\{\allocs\in\polytope\;\middle|\;x_i(\sigs)\in \{0,1\}, \forall i \in [n],\sigs \in [k]^n \right\},
\end{equation*}
i.e., the subset of feasible integer points obtained by imposing integrality constraints on truthful polytope.
\end{definition}

Note that our truthfulness notion for randomized mechanism is defined in expectation over the randomness of the mechanism: at equilibrium, truth-telling is a best-response which maximizes an agent's expected utility. However, a very informed agent who has access to the internal randomness of the allocation rule will know whether or not they will be selected, and it may not be optimal for them to reveal their true signal. A stronger notion is that of \emph{universally truthfulness}, when a randomized mechanism is a lottery over truthful deterministic mechanisms, which cannot be manipulated, even by agents who have access to the internal randomness of the allocation rule. Formally, the set of universally truthful mechanisms is the convex hull of $\intpolytope$, and the two notions of truthfulness coincide if and only if the extreme points of $\polytope$ are integral.

%%%%%%%%%%%%%%%%%%%%%%%%%%%%%%%%%%%%%%%%%%%%%%%%%%%%%%%%%%%%%%%
%%%%%%%%%%%%%%%%%%%%%%%%%%%%%%%%%%%%%%%%%%%%%%%%%%%%%%%%%%%%%%%
\subsection{Approximation Ratio}
Our goal would be to allocate the good or chore to the most deserving agent (maximum value or minimum cost). However, this might not always be feasible, as the optimal allocation might not satisfy the monotonicity condition of \Cref{lem:truthful}, and hence cannot be implemented truthfully. We measure the efficiency of a (truthful) mechanism as the worst ratio between its performance and that of the optimal (non-truthful) solution. When allocating a \emph{good} the approximation ratio is
\begin{equation*}
\Rv(\rhos, \allocs) :=
\max_{\sigs\in[k]^n} 1/\avg{\allocs(\sigs)}{\rhos(\sigs)}
=
\max_{\sigs\in[k]^n} \frac{\max_{i\in[n]} v_i(\sigs)}{\sum_{i\in[n]} x_i(\sigs) \cdot v_i(\sigs)} \geq 1.
\tag{$\Rv$}\label{eq:Rv}
\end{equation*}
When allocating a \emph{chore} the approximation ratio is
\begin{equation*}
\Rc(\rhos, \allocs) :=
\max_{\sigs\in[k]^n} \avg{\allocs(\sigs)}{1/\rhos(\sigs)}
=
\max_{\sigs\in[k]^n} \frac{\sum_{i\in[n]} x_i(\sigs) \cdot c_i(\sigs)}{\min_{i\in[n]} c_i(\sigs)}
\geq 1.
\tag{$\Rc$}\label{eq:Rc}
\end{equation*}
Abusing notations, we might write $\Rv(\vals, \allocs)$, $\Rc(\costs, \allocs)$, and $\Rdoptv$, or even $\Rv(\allocs)$ and $\Rc(\allocs)$, when the instance is clear from the context.
We observe that these quantities are closely related, through the following lemma. The main intuition is that approximate minimization is harder than approximate maximization, which can be shown using Jensen's convexity inequality.
\begin{lemma}\label{lem:ratios}
Given an allocation $\allocs$, we have
$\Rv(\rhos,\allocs) \leq \Rc(\rhos,\allocs)$,
with equality if $\allocs$ is deterministic, in which case we just write $R(\rhos,\allocs)$. 
\end{lemma}
\begin{proof}
The inequality holds by convexity of $x\mapsto 1/x$, using Jensen's inequality.
\end{proof}

Given performance ratios $\rhos$, we define the \emph{optimal approximation ratios}, for the value, cost and deterministic settings:
\begin{align*}
    \Rv^*\paras &:= \max_{\allocs\in \polytope} \Rv(\rhos, \allocs)\tag{$\Rv^*$}\\
    \qquad
    \Rc^*\paras &:= \max_{\allocs\in \polytope} \Rc(\rhos, \allocs)\tag{$\Rc^*$}\\
    \qquad
    \Rd^*\paras &:= \max_{\allocs\in \intpolytope} R(\rhos, \allocs) \tag{$\Rd^*$}
\end{align*}
Abusing notations, we might write $\Rv^*(\vals)$ and $\Rc^*(\costs)$, or even $\Rv^*$, $\Rc^*$ and $\Rd^*$, when the instance is clear from the context.
% \end{definition}
From \Cref{lem:ratios}, we obtain the following corollary.
\begin{corollary} \label{cor:ratios}
For all parameters $\paras$, we have that
$1\leq\Rv^*\paras \leq \Rc^*\paras \leq \Rd^*\paras.$
\end{corollary}

%%%%%%%%%%%%%%%%%%%%%%%%%%%%%%%%%%%%%%%%%%%%%%%%%%%%%%%%%%%%%%%
%%%%%%%%%%%%%%%%%%%%%%%%%%%%%%%%%%%%%%%%%%%%%%%%%%%%%%%%%%%%%%%
\subsection{Problems and Complexity}

\paragraph{Optimization problems.} The main problems we consider are optimization problems. More formally, an instance can be described with two parameters $\paras$, where $\sigmas$ characterizes the feasible solutions, and $\rhos$ induces the measure function.
We define the optimization problems $\PbV$, $\PbC$ and $\PbD$, which respectively ask to compute an allocation rule $\allocs$ which achieves $\Rv^*\paras$, $\Rc^*\paras$ and $\Rd^*\paras$, defined in the previous section. 

\paragraph{Query problems.}
Our optimization problems output allocation rules whose description can potentially be very large. In mechanism design, we mostly care about evaluating the allocation rule $\allocs$ at the reported signal profile $\sigs$. To explore the computational complexity of mechanisms, we introduce the query problems $\PbV_\gamma$, $\PbC_\gamma$ and $\PbD_\gamma$, which asks to answer evaluation queries about a solution $\allocs$ of measure $\leq \gamma$, for some $\gamma \geq 1$, while being consistent across queries.

\paragraph{Gap problems.}
In order to provide NP-hardness results, it is necessary to work with decision problems. Thus, we define the gap variants $(\alpha,\beta)$-\PbV, $(\alpha,\beta)$-\PbC and $(\alpha,\beta)$-\PbD, which ask whether the optimal solution has measure at most $\alpha$ or more than $\beta$, for some $1\leq \alpha \leq \beta$, under the promise that we are not in the intermediate case (the algorithm is allowed to fail or not to terminate in that case). When an $(\alpha,\beta)$-gap problem is NP-Hard, the corresponding optimization problem is hard to approximate within a factor $\leq\beta/\alpha$, as any approximate algorithm would distinguish between the $\leq\alpha$ and $>\beta$ cases. An approximation algorithm for our optimization problems might be confusing as the objective functions are approximation ratio themselves, however we reiterate the distinction between algorithms which solve the computational tasks and their outputs which are mechanisms (allocation rule and payment function).

\begin{table}[]
    \centering
    \begin{tabular}{|c|c|c|c|}
        \hline
        Problem & Output & Promise \\
        \hline\hline
        \PbV & $\allocs\in\polytope$ minimizing $\Rv(\rhos, \allocs)$ & --  \\
        \PbC & $\allocs\in\polytope$ minimizing $\Rc(\rhos, \allocs)$ & --  \\
        \PbD & $\allocs\in\intpolytope$ minimizing $R(\rhos, \allocs)$ & --  \\
        \hline
        $\PbV_\gamma$  & Oracle for $\allocs\in \polytope$ with $\Rv(\rhos, \allocs) \leq \gamma$ &  $\Rv^*\paras \leq \gamma$ \\
        $\PbC_{\gamma}$  & Oracle for $\allocs\in \polytope$ with $\Rc(\rhos, \allocs) \leq \gamma$ &  $\Rc^*\paras \leq \gamma$ \\
        $\PbD_{\gamma}$  & Oracle for $\allocs\in \intpolytope$ with $R(\rhos, \allocs) \leq \gamma$ &  $\Rd^*\paras \leq \gamma$ \\
        \hline
        $(\alpha,\beta)$-\PbV  & $\ind{\Rv^*\paras\leq \alpha}$ & $\Rv^*\paras \leq \alpha$ or $\Rv^*\paras > \beta$ \\
        $(\alpha,\beta)$-\PbC  & $\ind{\Rc^*\paras\leq \alpha}$ & $\Rc^*\paras \leq \alpha$ or $\Rc^*\paras> \beta$ \\
        $(\alpha,\beta)$-\PbD  & $\ind{\Rd^*\paras\leq \alpha}$ & $\Rd^*\paras \leq \alpha$ or $\Rd^*\paras > \beta$ \\
        \hline
    \end{tabular}
    \caption{Summary of the computational problems we consider. If the problem has a promise, an algorithm solving it is allowed to fail or not terminate on instances which do not satisfy it.}
    \label{tab:placeholder}
\end{table}

\paragraph{Time complexity.}
When computing the time complexity of an algorithm solving one of our computational problems, we assume that the input is a bit-string describing the cost or value functions. Therefore, the input has size $N = nk^nb$, where $b$ denotes the maximum number of bits used to represent a value or a cost. The dependency in $b$ is necessary as some of our algorithm will rely on solving linear programs for which all known algorithms perform a number of arithmetic operations which depend on $b$ \cite{khachiyan1979polynomial,karmarkar1984new,tardos1986strongly}.

\paragraph{Query complexity.} To provide lower bounds on the complexity of our computational problems, we consider the setting where the algorithm has oracle access to the value and cost functions, and can perform arbitrarily many arithmetic operations. We measure the query complexity as the number of queries necessary to distinguish between several instances which have different outputs. This corresponds to the complexity in the algebraic decision tree model \cite{ben1983lower}.

%%%%%%%%%%%%%%%%%%%%%%%%%%%%%%%%%%%%%%%%%%%%%%%%%%%%%%%%%%%%%%%
%%%%%%%%%%%%%%%%%%%%%%%%%%%%%%%%%%%%%%%%%%%%%%%%%%%%%%%%%%%%%%%
%%%%%%%%%%%%%%%%%%%%%%%%%%%%%%%%%%%%%%%%%%%%%%%%%%%%%%%%%%%%%%%
\section{Main Ideas of Our Techniques}
\label{sec:ideas}

In this section, we highlight the main technical contributions of our paper. For the simplicity of exposition, we will only discuss the value setting and problem \PbD.
The complete proofs of our results, including problem \PbV and \PbC, can be found in \Cref{sec:positive-results,sec:negative-results}.

A first useful trick is to reduce the optimization problem to its query variant, by running a binary search on the optimal approximation ratio $\Rd^*(\vals)$.
\begin{lemma}\label{lem:binary-search}
    Given $\gamma\geq 1$, if one can answer in $g(N)$ time all the $\PbD_\gamma$ queries, then we can solve the optimization problem $\PbD$ in $O(g(N) \log N)$ time.
\end{lemma}
\begin{proof}
    To compute $\Rdoptv$ and the corresponding deterministic mechanism, we will run a binary search: for all $\gamma\geq 1$ we can decide whether $\Rdoptv \leq \gamma$ by computing the answer to all $\PbD_\gamma$ queries, and checking if the resulting the solution $\allocs$ is correct. Because $\Rdoptv$ must be equal to $1/\rho_i(\sigs)$ for some $i\in[n]$ and $\sigs\in[k]^n$, there are at most $nk^n$ possible values, and our binary search will go over at most $O(\log N)$ values of $\gamma$.
\end{proof}

%%%%%%%%%%%%%%%%%%%%%%%%%%%%%%%%%%%%%%%%%%%%%%%%%%%%%%%%%%%%%%%
%%%%%%%%%%%%%%%%%%%%%%%%%%%%%%%%%%%%%%%%%%%%%%%%%%%%%%%%%%%%%%%
\subsection{Reduction to Bipartite Matching}

First consider the case where agents have binary signals, which was studied in \cite{EdenFFG18,AmerT21} for restricted domains of value functions. We propose a new approach, reducing the problem $\PbD _\gamma$ to finding a matching in a bipartite graph.
More details about this construction can be found in \Cref{sssec:hypergraph,sssec:two-signals}.

To build our graph, we first need to define at each signal profile $\sigs$ the set $A_\gamma(\sigs)$ of acceptable agents (i.e., agents who can be selected given the ratio $\gamma$) and the set $C(\sigs)$ of constrained agents  (i.e., agents whose selection is constrained by the monotonicity condition).
\begin{align*}
\forall \sigs\in[2]^n,\qquad
C(\sigs) &:= \{i\in[n]\;|\;\exists s_i'\in[2], (s_i, s_i')\in \sigma_i(\osigs)\},\\
A_\gamma(\sigs) &:= \{i\in[n]\;|\;\rho_i(\sigs) \geq 1/\gamma\}.
\end{align*}
The signal profiles are divided into two classes. A signal profile $\sigs$ is called \emph{must-match} if $A_{\gamma}(\sigs)\subseteq C(\sigs)$, that is, if all acceptable agent are constrained; otherwise, it is \emph{may-match}. The set of must-match profiles is denoted by $M_{\gamma}$. We construct the graph $G_\gamma = (V, E)$ as follows:
\begin{itemize}
    \item each vertex represents a signal profile, that is, $V = [2]^n$,
    \item for all $i\in[n]$, $\sigs\in[2]^{n-1}$, and $s_i'\neq s_i$, add to an edge from $\sigs$ to $(s_i',\osigs)$ if:
    \begin{itemize}
        \item $(s_i,s_i')\in \sigma_i(\osigs)$, and
        \item $i\in A_\gamma(\sigs)\cap A_\gamma(s_i',\osigs)$, and
        \item $\sigs\in M_\gamma$, 
    \end{itemize}
\end{itemize}
First, $G_\gamma$ is bipartite as one can partition vertices $\sigs$ into two sides based on the parity of $\sum_{i\in [n]} s_i$, and no edge can exist between two vertices on the same side. %\textcolor{blue}{Reviewer A said sth, but I don't think it's necessary.}
Second, notice that a vertex always connects a must-match vertex $\sigs$ to a may-match vertex $(s_i',\osigs)$, as $i\in A_\gamma(s_i',\osigs)\setminus C(s_i',\osigs)$.

\begin{figure}[h!]
    \centering
    \begin{tikzpicture}[xscale=4.5, yscale=2,>=stealth]
        \small
        \node[draw, fill=black!20!white, rounded corners] (s000) at (0,3) {$\rho$(1,1,1) = (\textbf{0.8}, 1.0, \st{0.2})};
        \node[draw, rounded corners] (s100) at (-1,2) 
        {$\rho$(2,1,1) = (\textbf{0.7}, 1.0, 0.9)};
        \node[draw, fill=black!20!white, rounded corners] (s010) at (0,2) 
        {$\rho$(1,2,1) = (\textbf{0.5}, \st{0.4}, \textbf{1.0})};
        \node[draw, fill=black!20!white, rounded corners] (s001) at (1,2)
        {$\rho$(1,1,2) = (1.0, \textbf{0.9}, \st{0.3})};
        \node[draw, rounded corners] (s110) at (-1,1)
        {$\rho$(2,2,1) = (\textbf{0.7}, 1.0, 1.0)};
        \node[draw, fill=black!20!white, rounded corners] (s101) at (0,1) 
        {$\rho$(2,1,2) = (\st{0.2}, \textbf{1.0}, \st{0.1})};
        \node[draw, rounded corners] (s011) at (1,1)
        {$\rho$(1,2,2) = (0.7, \textbf{0.6}, \textbf{1.0})};
        \node[draw, rounded corners] (s111) at (0,0)
        {$\rho$(2,2,2) = (1.0, \textbf{0.6}, \st{0.2})};
        \draw[->] (s000.-90) -- (s100.90);
        \draw[->,dotted] (s000.-90) -- (s010.90);
        \draw[->,dotted] (s000.-90) -- (s001.90);
        \draw[->,dashed] (s100.-90) -- (s110.90);
        \draw[->,dotted] (s100.-90) -- (s101.90);
        \draw[->] (s010.-90) -- (s110.90);
        \draw[->] (s010.-90) -- (s011.90);
        \draw[->] (s001.-90) -- (s011.90);
        \draw[->,dotted] (s001.-90) -- (s101.90);
        \draw[->] (s101.-90) -- (s111.90);
        \draw[->,dotted] (s110.-90) -- (s111.90);
        \draw[->,dashed] (s011.-90) -- (s111.90);
    \end{tikzpicture}
    \caption{Bipartite graph $G_\gamma$ with $\gamma = 2$.
    First, at each vertex $\sigs$ we strike-out non acceptable agents $i\notin A_\gamma(\sigs)$, for which $\rho_i(\sigs) < 1/\gamma$. 
    Second, we draw the edges of the cube, oriented using $\sigmas$ (in the special case of increasing value functions), dotted if the corresponding agent is not acceptable at one of the two endpoints.
    Third, we color gray the must-match vertices $\sigs$ with $A_\gamma(\sigs)\subseteq C(\sigs)$, for which all acceptable agents correspond to outgoing edges.
    Finally we dash edges which are not incident to a must-match vertex. The final graph is the set of plain edges.}
\end{figure}
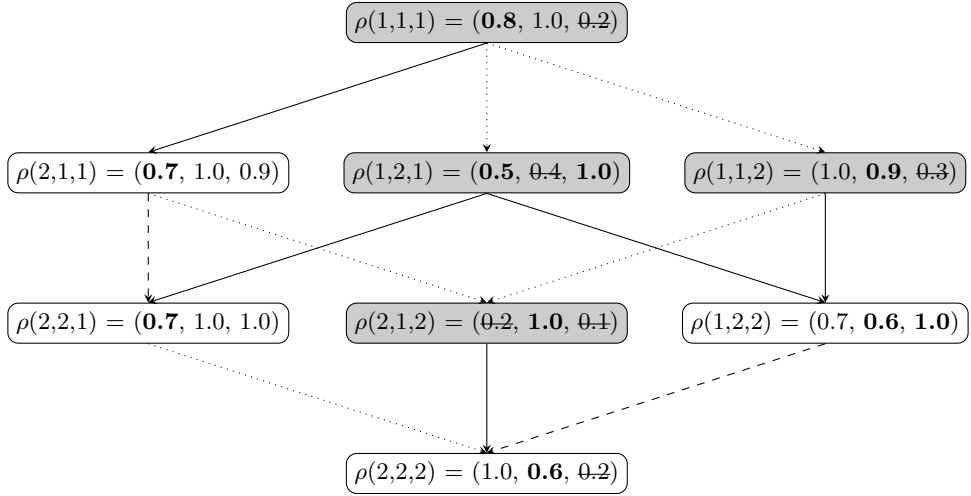

Finally, we show that there exists a $\gamma$-approximate deterministic truthful mechanism if and only if there exists a matching in $G_\gamma$ with $|M_{\gamma}|$ edges, which can be found by computing a maximum cardinality matching.
\begin{itemize}
    \item Assume that there exists a matching in $G_\gamma$ with $|M_{\gamma}|$ edges. If an edge in the matching connects $\sigs$ to $(s_i',\osigs)$, we select bidder $i$ in both profiles. As every edge must connect one must-match vertices, and there are $|M_{\gamma}|$ edges in the matching, the remaining vertices (if any) must be may-match. We select an arbitrary bidder of $A_{\gamma}(\sigs)\setminus C(\sigs)$ in the remaining vertices $\sigs$. It is easy to check that such an allocation satisfies monotonicity and approximation constraints.
    \item Suppose now that there is a $\gamma$-approximate deterministic truthful mechanism. If $\sigs$ is must-match, and $i$ is selected at $\sigs$, then we match $\sigs$ to $(s_i', \osigs)$ with $s_i'\neq s_i$. Note that the edge exists as truthfulness implies that $i$ is also selected at $(s_i', \osigs)$. Moreover, observe that each edge connects exactly one must-match vertex, thus we obtain a matching with $|M_\gamma|$ edges.
\end{itemize}

%%%%%%%%%%%%%%%%%%%%%%%%%%%%%%%%%%%%%%%%%%%%%%%%%%%%%%%%%%%%%%%
%%%%%%%%%%%%%%%%%%%%%%%%%%%%%%%%%%%%%%%%%%%%%%%%%%%%%%%%%%%%%%%
\subsection{Reduction to Boolean Satisfiability}
Next, we consider the special case with two agents, which was studied in \cite{EdenFFG18} for restricted domains of valuation functions. We propose a new approach, reducing $\PbD_\gamma$ to the satisfiability of a boolean formula. More details about this construction can be found in \Cref{sssec:SAT,sssec:two-agents}.

When $n=2$, observe that a deterministic allocation rule is completely characterized by the variables $x_1(\sigs)$ with $\sigs\in[k]^2$, as $x_2(\sigs) = 1-x_1(\sigs)$. Because the variable are in $\{0,1\}$, we can encode the monotonicity and ratio constraints with the following boolean formula.
\begin{align*}
    \forall s_2\in[k], \forall (s_1, s_1')\in\sigma_1(s_2),\qquad
    &
    x_1(s_1,s_2) \Rightarrow x_1(s_1',s_2),&(\text{i.e., }x_1(s_1, s_2) \leq x_1(s_1', s_2))
    \\
    \forall s_1\in[k], \forall (s_2, s_2')\in\sigma_2(s_1),\qquad
    &
    x_1(s_1,s_2')\Rightarrow x_1(s_1,s_2) ,&(\text{i.e., }x_2(s_1, s_2) \leq x_2(s_1, s_2'))
    \\
    \forall \sigs\in[k]^2 \text{ s.t. }\rho_1(\sigs) < 1/\gamma,\qquad
    &
    \neg x_1(\sigs),
    &
    (\text{i.e., }x_1(s_1, s_2) = 0)
    \\
    \forall \sigs\in[k]^2 \text{ s.t. }\rho_2(\sigs) < 1/\gamma,\qquad
    &
    x_1(\sigs).
    &
    (\text{i.e., }x_2(s_1, s_2) = 0)
\end{align*}
The formula could be rewritten as a conjunction of clauses of size $2$, which corresponds to the 2SAT problem. However, we will keep this formulation, which corresponds to the graph used to solve 2SAT using the strongly connected components \cite{AspvallPT79}. More precisely, we build a directed graph $G(\sigmas)$ over the set of signal profiles $[k]^2$, adding as edges all the implications from the boolean formula, as illustrated in \Cref{fig:2SAT}.
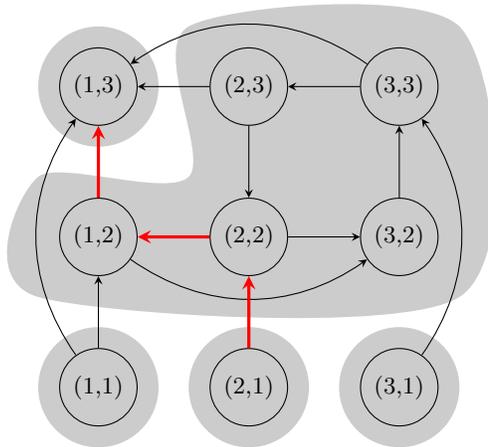
\begin{figure}[h!]
    \centering
    \begin{tikzpicture}[scale=2,>=stealth]
        \small
        \begin{scope}[black!20!white]
        \fill plot[smooth cycle] coordinates {(0.6,1.6)(3.4,1.6)(3.4,3.4)(1.6,3.4)(1.6,2.4)(0.6,2.4)};
        \fill (1,3) circle (.4cm);
        \fill (1,1) circle (.4cm);
        \fill (2,1) circle (.4cm);
        \fill (3,1) circle (.4cm);
        \end{scope}
        \foreach \i in {1,2,3}
        \foreach \j in {1,2,3}
        \node[draw,circle] (\i\j) at (\i,\j) {(\i,\j)};
        \begin{scope}[->]
        \draw (12) to [bend right=35] (32);
        \draw (32) -- (33);
        \draw (33) -- (23);
        \draw (23) -- (22);
        \draw[very thick,red] (22) -- (12);
        \draw (22) -- (32);
        \draw (31) to [bend right=35] (33);
        \draw[very thick,red] (21) -- (22);
        \draw (23) -- (13);
        \draw (33) to [bend right=35] (13);
        \draw (11) to [bend left=35] (13);
        \draw (11) -- (12);
        \draw[very thick,red] (12) -- (13);
        \end{scope}
        \normalsize
        \node at (5.5,2.5) {\begin{minipage}{5cm}
        Assume that $\gamma = 2$ and:
        \begin{itemize}
            \item $\rho_2(2,1) = 0.4 \leq 1/\gamma$, \\thus $x_1(2,1) = 1$;
            \item $\rho_1(1,3) = 0.3 \leq 1/\gamma$, \\thus $x_1(1,3) = 0$.
        \end{itemize}
        More precisely, the red chain exists for all $\gamma < 2.5 = 1/0.4$.
        \end{minipage}};
    \end{tikzpicture}
    \caption{Example of graph $G(\sigmas)$, obtained from the boolean formula with $n=2$ agents. For each $i\in[n]$ and $\osigs\in[k]$ the monotonicity constraints are given by the (partial) strict order $\sigma_i(\osigs)$. There exists a blocking chain of implications, represented in red. To compute efficiently the existence of a blocking chain, one could remove redundant edges that can be deduced using transitivity, and compute the strongly connected components, represented in gray.}
    \label{fig:2SAT}
\end{figure}

One can see that there is no satisfying assignment if and only if one can find a blocking chain of implications $1 = x_1(\sigs)\Rightarrow \dots \Rightarrow x_1(\sigs') = 0$, where the beginning and the end of the chain are constrained because of $\rhos$. This can be checked in polynomial time using depth-first search, and can be improved to $O(k^2)$ time if we are careful not to add too many edges while constructing the graph.

%%%%%%%%%%%%%%%%%%%%%%%%%%%%%%%%%%%%%%%%%%%%%%%%%%%%%%%%%%%%%%%
%%%%%%%%%%%%%%%%%%%%%%%%%%%%%%%%%%%%%%%%%%%%%%%%%%%%%%%%%%%%%%%
\subsection{Conflicting Pairs}

In the previous section, we discussed how to solve $\PbD_\gamma$ when $n=2$ using boolean satisfiability, which can be reformulated as a graph reachability problem. We now give a characterization of the optimal deterministic ratio $\Rd^*$, which can be computed directly without having to use the binary search reduction to $\PbD_\gamma$. More details about this construction can be found in \Cref{sssec:two-agents}.

A first step is to observe the previous reduction, while slowly decreasing $\gamma$ from $+\infty$ to $1$. At the beginning, there are no $\sigs\in[k]^2$ and $i\in[2]$ such that $\rho_i(\sigs) \leq 1/\gamma$, and thus there can be no blocking chain of implication. Then, we start to fix some variable $x_1(\sigs)$. Until we fix one last variable which creates a blocking chain of implications and blocks the existence of a satisfying assignment. Using this remark, we can exactly characterize the optimal deterministic ratio:
$$
\Rd^*\paras = \max_{\sigs\rightarrow\sigs'} \min\left(\frac{1}{\rho_2(\sigs)}, \frac{1}{\rho_1(\sigs')}\right),
$$
where $\sigs\rightarrow\sigs'$ denotes the reachability in the graph $G(\sigmas)$, and we call $(\sigs,\sigs')$ a \emph{conflicting pair}. Importantly, after sorting signals in $O(k^2\log k)$ we can compute this value in $O(k^2)$ time by using dynamic programming on the directed acyclic graph (DAG) of the strongly connected components of $G(\sigmas)$, together with a satisfying assignment. In \Cref{sssec:two-agents} we also give a characterization for the ratios $\Rv^*$ and $\Rc^*$ solving the problems \PbV, \PbC and \PbD in $O(k^2 \log k)$.

%%%%%%%%%%%%%%%%%%%%%%%%%%%%%%%%%%%%%%%%%%%%%%%%%%%%%%%%%%%%%%%
%%%%%%%%%%%%%%%%%%%%%%%%%%%%%%%%%%%%%%%%%%%%%%%%%%%%%%%%%%%%%%%
\subsection{Query Complexity}
\label{ssec:ideas-query}

In the previous sections, we gave (the intuition for) polynomial time algorithms solving \PbD in special cases. However, in mechanism design, one might only be interested in computing $\allocs(\sigs)$ at a specific signal profile $\sigs\in [k]^n$. Recall that the query problem $\PbD_\gamma$ exactly formalizes this possibility. Building on the previous two sections, we now prove \Cref{thm:query} in the special case where $n=2$, showing that one has to read $\Omega(k^2)$ values to compute the outcome at a specific signal profile. More details about the general construction can be found in \Cref{ssec:query}.

When $n=2$, we have seen that that the monotonicity condition of \Cref{lem:truthful} can be formulated using the directed graph $G(\sigmas)$. For simplicity, we will build an instance where value functions are increasing, which gives a graph $G(\sigmas)$ where $\sigs\rightarrow \sigs'$ if and only if $s_1 \geq s_1'$ and $s_2 \leq s_2'$.
We set $\varepsilon > 0$, and we fix a signal profile $\sigs$ where $s_1 = s_2 = 1+\lfloor k/2\rfloor$. Then, we draw a random signal profile $\sigs'\in[k]^2$:
\begin{itemize}
    \item if we have $\sigs'\rightarrow\sigs$, we set $\rho_2(\sigs') = \varepsilon$ and all other performance ratios at 1,
    \item if we have $\sigs\rightarrow\sigs'$, we set $\rho_1(\sigs') = \varepsilon$ and all other performance ratios at 1,
    \item otherwise, we set all performance ratios at 1.
\end{itemize}
Observe that there always exists a very simple allocation rule, which selects the same agent at all signal profiles, and achieves a ratio of $1$. Unfortunately, any allocation rule $\allocs$ solving $\PbD_1$ must have $x_1(\sigs) = 1$ if  $\sigs'\rightarrow\sigs$, and $x_2(\sigs) = 1$ if $\sigs\rightarrow\sigs'$. Therefore, in the worst case one has to query the values at $\Omega(k^2)$ signal profiles to decide in which situation we are. This construction can be formalized using the algebraic decision tree model \cite{ben1983lower}. As a side note, even a randomized exploration strategy to find $\sigs'$ would not improve the query complexity, as the adversary is not adaptive.

%%%%%%%%%%%%%%%%%%%%%%%%%%%%%%%%%%%%%%%%%%%%%%%%%%%%%%%%%%%%%%%
%%%%%%%%%%%%%%%%%%%%%%%%%%%%%%%%%%%%%%%%%%%%%%%%%%%%%%%%%%%%%%%
\subsection{NP-Hardness}

To conclude this section about the important ideas of our techniques, we give some intuition on the proof of \Cref{thm:nphard}, which states that $(1,\beta)$-\PbD is NP-Hard when $n=4$. More details about this construction can be found in \Cref{ssec:nphard}.

For simplicity, we will build an instance where value functions are increasing, which simplifies the monotonicity condition of \Cref{lem:truthful}:
$$
\forall i\in[n], \forall \osigs\in[k]^{n-1},\forall 1 \leq s_i \leq s_i' \leq k,\qquad x_i(s_i, \osigs) \leq x_i(s_i',\osigs).
$$

We prove the NP-hardness result through reduction from 1-in-3-SAT problem, which asks a boolean assignment to variables, with constraints requiring exactly one of three literals (variable or their negation) to be true. In the following, we illustrate the main ideas of the reduction through \Cref{fig:hardness-var,fig:hardness-clause,fig:hardness-connect}.
For convenience, we set a constant $\varepsilon \in(0,1/\beta)$.
To build some intuition, we set
$$\rhos(1,2,1,*) = (1,\varepsilon,1,\varepsilon),
\qquad
\rhos(1,1,2,*) = (1,1,\varepsilon,\varepsilon),
\qquad
\rhos(k,1,1,*) = (\varepsilon,1,1,\varepsilon).$$
First, because of the performance ratios, if $\allocs\in\intpolytope$ such that $\Rd(\allocs, \rhos) = 1$ then we have
$$
x_1(k,1,1,*) =x_4(k,1,1,*) = x_2(1,2,1,*) = x_4(1,2,1,*) = x_3(1,1,2,*) = x_4(1,1,2,*) = 0.
$$
If we decide to set $x_1(1,1,2,*) = 1$, then
\begin{align*}
x_1(1,1,2,*) = 1
&\quad\Rightarrow\quad
x_1(k,1,2,*) = 1
& (\text{by monotonicity})\\
&\quad\Rightarrow\quad
x_3(k,1,2,*) = 0
& (\text{sum of proba is 1}) \\
&\quad\Rightarrow\quad
x_3(k,1,1,*) = 0
& (\text{by monotonicity}) \\
&\quad\Rightarrow\quad
x_2(k,1,1,*) = 1
& (\text{sum of proba is 1}) \\
&\quad\Rightarrow\quad \dots\\
&\quad\Rightarrow\quad
x_1(1,1,2,*) = 1
\end{align*}
Thus, all these are equivalent, and if they hold we say that our gadget is true. Conversely, we say that our gadget is false if $x_1(1,2,1,*) = 1$ holds. This ``variable'' gadget, illustrated in \Cref{fig:hardness-var} is just one building block of our reduction from the 1-in-3-SAT problem.

In \Cref{fig:hardness-clause}, we combine three ``variable'' gadgets by having them intersect in one signal profile, requiring exactly one of the three literals to be true, which constitutes a ``1-in-3 clause'' gadget. Finally, in \Cref{fig:hardness-connect} we show how to connect ``clause'' and ``variable'' gadgets, requiring the values of variables and literals to be consistent.

A more detailed description can be found in \Cref{ssec:nphard}.

\begin{figure}[p!]%[htbp]
    \centering
    \begin{tikzpicture}[scale=.9,thick,z=-.66cm]
    \begin{scope}[blue!50!white]
    \begin{scope}[line width=3pt]
    \draw (0,0,1) -- (6,0,1);
    \draw[densely dashed] (6,0,1) -- (6,0,0);
    \draw (6,0,0) -- (6,1,0);
    \draw[densely dashed] (6,1,0) -- (0,1,0);
    \draw (0,1,0) -- (0,1,1);
    \draw[densely dashed] (0,1,1) -- (0,0,1);
    \end{scope}
    \fill (0,0,1) circle (3pt);
    \node[anchor=east] at (-.1,0,1) {$\rhos(1,1,2,*)=(1,1,\varepsilon,\varepsilon)$};
    \node[anchor=east] at (-.1,1.2,0) {$\rhos(1,2,1,*) = (1,\varepsilon,1,\varepsilon)$};
    \node[anchor=west] at (6,-.5,0) {$\rhos(k,1,1,*)=(\varepsilon,1,1,\varepsilon)$};
    \foreach \x/\y/\z in {0/0/1,6/0/0,0/1/0,0/0/1}
    \fill (\x,\y,\z) circle (3pt);
    \node[anchor=south] at (3,1,0) {$\neg a$};
    \node[anchor=north] at (3,0,1) {$a$};
    \end{scope}
    \draw[->] (0,0,0) -- (7,0,0);
    \draw[->] (0,0,0) -- (0,2,0);
    \draw[->] (0,0,0) -- (0,0,2);
    \node at (7.5,0,0) {$s_1$};
    \node at (0,2.5,0) {$s_2$};
    \node at (0,0,2.5) {$s_3$};
    \end{tikzpicture}
    \caption{Gadget for a variable $a$ located at $(s_2,s_3) = (1,1)$. We have $n=4$ agents, but we do not represent the fourth dimension (we set $\rhos$ to be constant over $s_4$). At three vertices, we set some performance ratio at $\varepsilon$ to prevent the items being allocated to, and we leave the other ratios at $1$.}
    \label{fig:hardness-var}
\end{figure}

\begin{figure}[p!]%[htbp]
    \centering
    \begin{tikzpicture}[scale=.9,thick,z=-.66cm]
    \draw[line width=3pt,blue!50!white,densely dashed] (0,0,0) -- (0,0,1) -- (0,-2,1) -- (7,-2,1) -- (7,-2,-3) -- (7,0,-3) -- (0,0,-3) -- (0,0,0);
    \draw[line width=3pt,blue!50!white] (0,0,-3) -- (0,0,1);
    \draw[line width=3pt,blue!50!white] (0,-2,1) -- (7,-2,1);
    \draw[line width=3pt,blue!50!white] (7,-2,-3) -- (7,0,-3);
    \draw[line width=5pt,white] (0,0,0) -- (0,1,0) -- (0,1,-2) -- (7,1,-2) -- (7,-1,-2) -- (7,-1,0) -- (0,-1,0) -- (0,0,0);
    \draw[line width=3pt,blue!50!white,densely dashed] (0,0,0) -- (0,1,0) -- (0,1,-2) -- (7,1,-2) -- (7,-1,-2) -- (7,-1,0) -- (0,-1,0) -- (0,0,0);
    \draw[line width=3pt,blue!50!white] (0,-1,0) -- (0,1,0);
    \draw[line width=3pt,blue!50!white] (0,1,-2) -- (7,1,-2);
    \draw[line width=3pt,blue!50!white] (7,-1,-2) -- (7,-1,0);
    \draw[line width=5pt,white] (0,0,0) -- (7,0,0) -- (7,0,-1) -- (7,2,-1) -- (-1,2,-1) -- (-1,2,0) -- (-1,0,0) -- (0,0,0);
    \draw[line width=3pt,blue!50!white,densely dashed] (0,0,0) -- (7,0,0) -- (7,0,-1) -- (7,2,-1) -- (-1,2,-1) -- (-1,2,0) -- (-1,0,0) -- (0,0,0);
    \draw[line width=3pt,blue!50!white] (-1,0,0) -- (7,0,0);
    \draw[line width=3pt,blue!50!white] (7,0,-1) -- (7,2,-1);
    \draw[line width=3pt,blue!50!white] (-1,2,-1) -- (-1,2,0);
    \fill[blue!50!white] (0,0,0) circle (5pt);
    \foreach \x/\y/\z in {0/0/-3,0/1/-2,-1/2/-1,-1/0/0,0/-1/0,0/-2/1,7/0/-1,7/-1/-2,7/-2/-3}
    \fill[blue!50!white] (\x,\y,\z) circle (3pt);
    \begin{scope}[blue!50!white]
    \end{scope}
    \draw[->] (0,0,0) -- (9.5,0,0);
    \draw[->] (0,0,0) -- (0,3.5,0);
    \draw[->] (0,0,0) -- (0,0,2);
    \node at (10,0,0) {$s_1$};
    \node at (0,4,0) {$s_2$};
    \node at (0,0,2.5) {$s_3$};
    \node[anchor=south east] at (0,0,0) {$\sigs$};
    \end{tikzpicture}
    \caption{Gadget for a clause $\ell_1\vee\ell_2\vee\ell_3$. We create one cycle per literal, which all intersect in one signal profile $\sigs$. We set $\rhos(\sigs) = (1,1,1,\varepsilon)$, and the choice of winner at $\sigs$ decides which of the three literal is true.}
    \label{fig:hardness-clause}
\end{figure}

\begin{figure}[p!]%[htbp]
    \centering
    \begin{tikzpicture}[scale=.9,thick,z=-.66cm]
    \begin{scope}[blue!50!white]
    \begin{scope}[line width=3pt]
    \draw (0,0,1) -- (8,0,1);
    \draw[densely dashed] (8,1,0) -- (0,1,0);
    \draw[densely dashed] (3,0,0) -- (3,0,1);
    \draw (3,0,0) -- (3,1,0);
    \draw (2,1,0) -- (2,1,1);
    \draw[densely dashed] (2,1,1) -- (2,0,1);
    \end{scope}
    \foreach \x/\y/\z in {3/0/0,2/1/0,2/0/1}
    \fill (\x,\y,\z) circle (3pt);
    \node[anchor=west] at (3.1,.25,0) {$\rhos(h+1,s_2,s_3',*)=(\varepsilon,1,1,\varepsilon)$};
    \node[anchor=north] at (3,0,1) {$\rhos(h,s_2,s_3,*)=(1,1,\varepsilon,\varepsilon)$};
    \node[anchor=south] at (3,1,0) {$\rhos(h,s_2',s_3',*)=(1,\varepsilon,1,\varepsilon)$};
    \end{scope}
    \draw[->] (0,0,0) -- (9,0,0);
    \draw[->] (0,0,0) -- (0,2,0);
    \draw[->] (0,0,0) -- (0,0,2);
    \node at (9.5,0,0) {$s_1$};
    \node at (0,2.5,0) {$s_2$};
    \node at (0,0,2.5) {$s_3$};
    \end{tikzpicture}
    \caption{Gadget for a XOR-connector at height $h$. It connects two horizontal lines $(s_2, s_3)$ and $(s_2', s_3')$ such that $s_2 < s_2'$ and $s_3 > s_3'$, and such that no other horizontal line shares a coordinate.}
    \label{fig:hardness-connect}
\end{figure}

%%%%%%%%%%%%%%%%%%%%%%%%%%%%%%%%%%%%%%%%%%%%%%%%%%%%%%%%%%%%%%%
%%%%%%%%%%%%%%%%%%%%%%%%%%%%%%%%%%%%%%%%%%%%%%%%%%%%%%%%%%%%%%%
%%%%%%%%%%%%%%%%%%%%%%%%%%%%%%%%%%%%%%%%%%%%%%%%%%%%%%%%%%%%%%%
\section{Algorithms and Upper Bounds}
\label{sec:positive-results}

We now present the proofs of our positive results: polynomial time algorithms for \PbV, \PbC, and special cases of \PbD. This section is organized in two parts. First, in \Cref{subsec:general-formulation} we give formulations of the general settings using linear and combinatorial optimization problems, some formulations being tractable and some being NP-Hard. Then, in \Cref{ssec:special-cases}, we explore the special cases where our NP-Hard problems become solvable efficiently. 

%%%%%%%%%%%%%%%%%%%%%%%%%%%%%%%%%%%%%%%%%%%%%%%%%%%%%%%%%%%%%%%
%%%%%%%%%%%%%%%%%%%%%%%%%%%%%%%%%%%%%%%%%%%%%%%%%%%%%%%%%%%%%%%
\subsection{General Formulations}
\label{subsec:general-formulation}

In this section, we give general formulation of our optimization problems: \PbV and \PbC can be expressed as solutions of linear programs, while \PbD can be expressed either as a solution of an integer program, as a boolean satisfiability problem, or as a matching problem in a hypergraph.

%%%%%%%%%%%%%%%%%%%%%%%%%%%%%%%%%%%%%%%%%%%%%%%%%%%%%%%%%%%%%%%
\subsubsection{Linear Programming}
\label{sssec:LP}

Given parameters $\paras$, solving \PbV and \PbC consists in minimizing $\Rv(\rhos, \allocs)$ and $\Rc(\rhos, \allocs)$ over the truthful polytope $\allocs\in\polytope$. Using \Cref{eq:Rc}, the ratio $\Rc(\rhos, \allocs)$ is the maximum of linear functions (in $\allocs$), and thus is convex. The situation with $\Rv(\rhos, \allocs)$ is slightly different, but one can show that $\allocs \mapsto 1/\avg{\allocs(\sigs)}{\rhos(\sigs)}$ is a convex function for all $\sigs$, and thus $\Rv(\rhos, \allocs)$ is also convex using \Cref{eq:Rv}. As minimizing convex functions on convex polytope is tractable, this provides an intuition on why \PbV and \PbC are computable in polynomial time. Importantly, we can refine our formulation, and express these two problems as solutions of linear programs, which yields a proof of \Cref{thm:lp}.

\thmlp*
\begin{proof}
Using a folklore construction, we can express the cost and value ratios as linear objectives:
\begin{align*}
        \text{minimize} \quad& \alpha&(\Rc^*\paras)\\
        \text{such that} \quad& \alpha \geq \avg{\allocs(s)}{1/\rhos(s)}\qquad\forall \sigs\in[k]^n
        \\& \allocs \in \polytope,
\end{align*}
and 
\begin{align*}
    \text{maximize} \quad& \beta&(1/\Rv^*(\paras))\\
    \text{such that} \quad& \beta \leq \avg{\allocs(s)}{\rhos(s)}\qquad\forall \sigs\in[k]^n\\
    & \allocs \in \polytope.
\end{align*}
Both linear programs can be optimized in polynomial time, using the ellipsoid method \cite{khachiyan1979polynomial} or the interior points method \cite{,karmarkar1984new}, and produce mechanisms which realize the optimal approximation ratios.
\end{proof}
Importantly, the complexity of solving linear programs is ``weakly'' polynomial in the sense that it scales polynomial in the number $b$ of bits used to express coefficients. Strongly polynomial time algorithm exists for special cases of linear programming \cite{tardos1986strongly}, but no strongly polynomial time algorithm is known for linear programs such as ours where the constraint matrix contains large coefficients.

%%%%%%%%%%%%%%%%%%%%%%%%%%%%%%%%%%%%%%%%%%%%%%%%%%%%%%%%%%%%%%%
\subsubsection{Integer Programming}
\label{sssec:ILP}

In the previous section, we saw how to compute the optimal value and cost ratios with linear programming. However, using the same approach to solve \PbD would require solving an integer program which is hard in general. We next show that it is tractable when the truthful polytope $\polytope$ has integral extreme points.

\begin{proposition}
\label{propo:polytime_deter}
    If $\polytope$ has integral vertices, then one can solve \PbD in polynomial.
\end{proposition}
\begin{proof}
    To compute $\Rdoptv$ and the corresponding deterministic mechanism, we will run the binary search of \Cref{lem:binary-search}. To answer the $\PbD_\gamma$ queries, we solve the following linear program:
    \begin{align*}
        \text{maximize} \quad& \sum_{\sigs\in[k]^n} \sum_{i\in[n]} x_i(\sigs) \cdot \mathbb 1[\rho_i(\sigs) \geq 1/\gamma] &(\Rdoptv\leq \gamma)\\
        \text{such that} \quad&  \allocs \in \polytope.
    \end{align*}
    Because the objective is linear and $\polytope$ is an integral polytope, the maximum is reached at an integral vertex $\allocs$. Moreover, observe that the maximum is equal to $k^n$ if and only if there exists a deterministic allocation $\allocs\in\polytope$ which selects at each signal profile $\sigs$ an agent $i\in[n]$ such that $\rho_i(\sigs)\geq 1/\gamma$, that is, if and only if $ \Rd^*\paras\leq \gamma$.
\end{proof}
Beyond polynomial-time computability of the deterministic ratio, the integrality of the truthful polytope also implies that all $\allocs\in\polytope$ are universally truthful, which we recall is defined as a lottery over truthful deterministic mechanisms.
\begin{proposition}
    \label{propo:inte_uni_truth}If the truthful polytope $\polytope$ has integral vertices, then allocation rules in $\polytope$ satisfy the stronger notion of \emph{universal truthfulness}.
\end{proposition}
\begin{proof}
    Any allocation rule $\allocs\in\polytope$ can be expressed as a convex combination of extreme points of $\polytope$. Using the fact that the truthful polytope is integral, $\allocs$ can be implemented as a lottery over deterministic allocation rules, which can all be implemented truthfully.
\end{proof}

In \Cref{sssec:integrality} we will show that $\polytope$ always has integral vertices when $n = 2$ or $k = 2$, but may have fractional vertices when $n \geq 3$ and $k\geq 3$.

%%%%%%%%%%%%%%%%%%%%%%%%%%%%%%%%%%%%%%%%%%%%%%%%%%%%%%%%%%%%%%%
\subsubsection{Boolean Satisfiability}
\label{sssec:SAT}

In the previous subsection, we formulated $\PbD_\gamma$ as maximizing a linear objective over $\intpolytope$, which is tractable if the related polytope $\polytope$ has integral vertices. Because in the deterministic case the variables $x_i(\sigs)$ are in $\{0,1\}$, one can encode the probability, monotonicity and ratio constraints as a boolean formula. Thus, $\PbD_\gamma$ can also be expressed as a boolean satisfiability problem.

More formally, in the SAT problem, each constraint is a \emph{clause} (disjunction) in which one of the \emph{literals} (variable or its negation) must be true. The resulting formula is the conjunction of all clauses, which is said to be in \emph{conjunctive normal form} (CNF). It is satisfiable if there exists a boolean assignment of the variable which makes the formula evaluate to true.

\begin{lemma}\label{lem:satisfiability}
Given parameters $\paras$ and $\gamma\geq 1$, the existence of a mechanism $\allocs\in\intpolytope$ with approximation ratio $R(\rhos,\allocs) \leq \gamma$ is equivalent with the satisfiability of the following formula:
\begin{align*}
    \forall\sigs\in[k]^n,\qquad
    &
    x_i(\sigs)\vee \dots \vee x_n(\sigs),
    \\ 
    \forall\sigs\in[k]^n,\forall i,j\in[n]\text{ s.t. }i\neq j,\qquad
    & 
    \neg x_i(\sigs) \vee \neg x_j(\sigs),
    \\
    \forall i\in[n],\forall \osigs\in[k]^{n-1}, \forall (s_i, s_i')\in\sigma_i(\osigs),\qquad
    &
    \neg x_i(s_i,\osigs)\vee x_i(s_i',\osigs),
    \\
    \forall \sigs\in[k]^n,\forall i\in[n]\text{ s.t. }\rho_i(\sigs) < 1/\gamma,\qquad
    &
    \neg x_i(\sigs).
\end{align*}
\end{lemma}
\begin{proof}
Observe that first two clauses correspond to the probability constraint in \Cref{eq:proba}, the third clause corresponds to the monotonicity constraint in \Cref{eq:monotonicity}, and the fourth clause imposes constraints on the approximation ratio.
\end{proof}

Unfortunately, the problem of finding a satisfying assignment for a CNF formula is NP-Hard in general. However, observe that when $n=2$ all clauses have size at most $2$, which corresponds to the 2-SAT problem, which can be solved in polynomial time. We will use this remark in \Cref{sssec:two-agents}.

%%%%%%%%%%%%%%%%%%%%%%%%%%%%%%%%%%%%%%%%%%%%%%%%%%%%%%%%%%%%%%%
\subsubsection{Perfect Matching in Hypergraphs}
\label{sssec:hypergraph}

Finally, we introduce in this section the third reformulation of $\PbD_\gamma$, as the problem of finding a perfect matching in a hypergraph. More formally, a hypergraph is described by a set of \emph{vertices} $V$, and a set of \emph{hyperedges} $E$, where each hyperedge $e\in E$ is a subset of vertices $e\subseteq V$. A matching is a collection of hyperedges $M\subseteq E$, such that each vertex is contained in at most one edge. The matching is perfect if each vertex is contained in exactly one hyperedge.
\begin{lemma}\label{lem:hypergraph}
Given parameters $\paras$ and $\alpha \geq 1$, the existence of a mechanism $\allocs\in\intpolytope$ with approximation ratio $R(\rhos, \allocs) \leq \alpha$ is equivalent with the existence of a perfect matching in the following hypergraph:
$$
V := [k]^n\qquad\text{and}\qquad
E := \bigcup_{i\in[n]}
\left\{
e\in E_i\;|\; \forall \sigs\in e, \rho_i(\sigs) \geq 1/\alpha
\right\}$$
where for all $i\in [n]$ we have
$$
    E_i = \left\{\{(s_i, \osigs)\;|\; s_i \in S\}\;\middle|\;\begin{array}{l}
    S\subseteq [k]\text{ and }\osigs\in [k]^{n-1}\text{ such that}\\
    \forall (s_i, s_i')\in \sigma_i(\osigs), \; s_i\in S \Rightarrow s_i'\in S\end{array}\right\}
$$
\end{lemma}
\begin{proof}
    Given a perfect matching $M\subseteq E$, for each $e\in M$ such that $e\in E_i$ we will set $x_i(\sigs) = 1$ for all $\sigs\in e$. One has to be careful when $e$ is in two sets $E_i$ and $E_j$ with $i\neq j$, in which case it holds that $e = \{\sigs\}$ and setting either $x_i(\sigs) = 1$ or $x_j(\sigs) = 1$ will work.

    First, observe that because the matching is perfect, for each $\sigs\in[k]^n$ there exists a unique $i\in[n]$ such that $x_i(\sigs) = 1$, and thus $\allocs$ satisfies \Cref{eq:proba}. Second, by construction of $E_i$, the resulting allocation rule satisfy the monotonicity condition of \Cref{eq:monotonicity}, and  is therefore truthful. Finally, by construction of $E$, for all $i\in[n]$ and $\sigs\in[k]^n$ such that $x_i(\sigs) = 1$ we have $\rho_i(\sigs) \geq 1/\alpha$, and thus $R(\rhos, \allocs) \leq \alpha$.
\end{proof}

This construction might seem overly complicated, as computing a perfect matching in a hypergraph is in general NP-Hard. However, observe that the hyperedges of the graph have size at most $k$, which corresponds to a standard graph when $k=2$. We will use this remark in \Cref{sssec:two-signals}.

%%%%%%%%%%%%%%%%%%%%%%%%%%%%%%%%%%%%%%%%%%%%%%%%%%%%%%%%%%%%%%%
%%%%%%%%%%%%%%%%%%%%%%%%%%%%%%%%%%%%%%%%%%%%%%%%%%%%%%%%%%%%%%%
\subsection{Special Cases with Efficient Algorithms}
\label{ssec:special-cases}

In \Cref{subsec:general-formulation}, we provided several generic reformulation of our optimization problem, using linear programming, boolean satisfiability and hypergraph matching, most of which being hard to solve in the general case. In this section, we explore special case, and show how to exploit these reformulations to design efficient algorithms.

%%%%%%%%%%%%%%%%%%%%%%%%%%%%%%%%%%%%%%%%%%%%%%%%%%%%%%%%%%%%%%%
\subsubsection{Conditions for Integrality}
\label{sssec:integrality}
In \Cref{sssec:ILP} analysis, we established that the integrality of the truthful polytope entails both polynomial-time computability of the deterministic ratio and universal truthfulness.
We now turn to the special cases in which integrality holds.

We first recall that if there exists a tie in values (or costs), i.e., if there exists $i$, $\sigs_{-i}$, $s_i$ and $s_i^{\prime}$ such that $v_i(s_i,\sigs_{-i})=v_i(s_i^{\prime},\sigs_{-i})$ (or $c_i(s_i,\sigs_{-i})=c_i(s_i^{\prime},\sigs_{-i})$), then the induced strict order is not total, as neither $(s_i,s_i^{\prime})$ nor $(s_i^{\prime},s_i)$ belongs to $\sigma_i(\sigs_{-i})$. Nevertheless, we claim that restricting attention to tie-free instances is without loss of generality. 
Since any strict partial order can be extended to a strict total order \cite{szpilrajn1930extension}, any set of strict orders $\sigmas$ could be extended to $\tilde{\sigmas}$ which is total, we denote this relation as $\sigmas\subset\tilde{\sigmas}$.

\begin{lemma}
    \label{lm:tie_free_TUM} 
    Truthful polytope $\polytope$ has integral vertices if $\polytopetilde$ has integral vertices for collection of strict total order $\tilde{\sigmas}$ such that $\sigmas \subset \tilde{\sigmas}$.
\end{lemma}
\begin{proof}
To facilitate the analysis, we introduce the notation $\polyver$ to denote the set of vertices of the truthful polytope $\polytope$. 
We state that 
$$
\polyver \subset \bigcup_{\tilde{\sigmas}:\,\sigmas \subset \tilde{\sigmas}} \polyvertilde,
$$
which directly implies the proposition. 
The statement holds for the following reason. 
Take any $\allocs \in \polyver$. 
Then there exists some $\tilde{\sigmas}$ such that $\sigmas \subset \tilde{\sigmas}$ and $\allocs \in \polytopetilde$. 
Indeed, we can extend each $\sigma_i(\sigs_{-i})$ by adding all pairs $(s_i,s_i')$ 
that satisfy both of the following conditions:
neither $(s_i,s_i')$ nor $(s_i',s_i)$ belongs to $\sigma_i(\sigs_{-i})$, 
and $x_i(s_i,\sigs_{-i}) < x_i(s_i',\sigs_{-i})$.
By arbitrarily extending the resulting $\sigmas$ to a total set of strict orders $\tilde{\sigmas}$, we obtain $\sigmas \subset \tilde{\sigmas}$. One can easily verify that $\allocs \in \polytopetilde$. 

Note that the constraints defining $\polytopetilde$ can be equivalently written so as to include all constraints of $\polytope$, which, together with the fact that $\allocs \in \polytopetilde$, shows that $\allocs \in \polyvertilde$. 
This follows directly from the standard algebraic characterization of vertices \cite[Proposition 2.9]{bertsimas1997introduction}.
\end{proof}
\begin{proposition}
    If $n = 2$ or $k = 2$, then $\polytope$ has only integral vertices.
    \label{propo:TUM}
\end{proposition}

\begin{proof}
    By \Cref{lm:tie_free_TUM}, we may, without loss of generality, assume that $\sigmas$ is total. In what follows, we work under this assumption.

    We will show that the constraint matrix of $\polytope$ is totally unimodular (TU), which, together with the fact that the vector of constraints' constants is integral, implies that the extreme points are integral \cite{schrijver1998theory}. For all $\allocs \in \mathbb R_{\geq 0}^{[n]\times [k]^n}$ we have that $\allocs\in\polytope$ if and only if
    $$
    \left(\begin{array}{ccc}
    \\
    \qquad &A& \qquad\\
    \\\hline\\
    \quad &-A& \qquad\\
    \\\hline\\
    \qquad &B&\qquad\\\\
    \end{array}\right)
    \cdot\allocs
    \quad\leq\quad 
    \left(\begin{array}{c}
    1\\\vdots\\ 1%\textcolor{blue}{-1?}
    \\\hline
    -1\\\vdots\\ -1%\textcolor{blue}{-1?}
    \\\hline
    0\\\vdots\\ 0
    \end{array}\right),
    $$
    where $A$ contains the probability constraint $\sum_i x_i(\sigs) \leq 1$ for all $\sigs$, and where $B$ contains the monotonicity constraints $x_i(s_i,\sigs_{-i}) - x_i(s_i^{\prime},\osigs) \leq 0$ for all pairs $(s_i,s_i^{\prime})$, where $s_i$ is the predecessor of $s_i^{\prime}$ in the total strict order $\sigmas_i(\osigs)$.
    
    First, duplicating rows and/or columns, or changing their sign preserve the total unimodularity of the matrix. Hence, it is enough to prove that the matrix stacking $A$ and $B$ is totally unimodular.
    Using \cite[Theorem 19.3]{schrijver1998theory}, a matrix is TU if and only if each collection $C$ of columns can be partitioned into classes $C_1$ and $C_2$ such that the sum of the columns in $C_1$, minus the sum of the columns in $C_2$, is a vector with entries $\{-1,0,1\}$ only. The same property holds for rows since the transpose of a TU matrix is also TU.
    \begin{itemize}
        \item If $n = 2$, then for every collection of columns $C$, we set $C_1$ to be the columns of type $x_1(\sigs)$, and we set $C_2$ to be the columns of type $x_2(\sigs)$. For every row of $A$, there are at most two non-zero entries, both equal to $1$, one in each set $C_1$ and $C_2$. For every row of $B$, there are at most two non-zero entries, equal to $1$ and $-1$, both in the same set $C_1$ or $C_2$.
        \item If $k = 2$, then for every collection of rows $R$, we assign the row $\sigs$ of $A$ to class $R_1$ if $\sum_j s_j$ is even, and to class $R_2$ otherwise. 
        Moreover, we assign the row $(i, \sigs)$ of $B$ to the same class as the row $\sigs$ of $A$ if the coefficient of $x_i(\sigs)$ in $B$ is $-1$, and to the opposite class if the coefficient is $1$. Note that there exist two distinct pairs $(i,\sigs)$ and $(i,\sigs')$ that correspond to the same row of $B$, yet one can verify that no conflicts arise in this classification. Then, each column $(i,\sigs)$ has at most $2$ non-zero entries, one in $A$ equal to $1$, and one in $B$ equal to either $-1$ or $1$. By construction, these two entries have the same sign if and only if they are in the different classes. 
    \end{itemize}
    In both cases, the resulting vector has coefficients in $\{-1, 0, 1\}$, proving that the matrix is TU.
\end{proof}

The above proposition, together with \Cref{propo:polytime_deter}, directly implies the following corollary.

\begin{corollary}
    \label{cor:poly_deter_n2k2}
    If $n=2$ or $k=2$, one can solve \PbD in polynomial time.
\end{corollary}

We complement the result of \Cref{propo:TUM} by showing that integrality does not hold if $n \geq 3$ and $k\geq 3$, even when $\sigmas$ is total and for any given $i$, $\sigma_{i}(\osigs)$ induces the same order for all $\osigs$.

\begin{proposition}
    If $n \geq 3$ and $k \geq 3$, then $\polytope$ may have some fractional vertices.
\end{proposition}
\begin{proof}
    We consider the case where valuation functions are strictly increasing, that is, where $\sigma_i(\osigs) = \{(s_i, s_i')\;|\;1 \leq s_i < s_i' \leq k\}$ for all $i\in[n]$ and $\osigs$.
    In \Cref{fig:fractional}, we plot two vertices of with $n=3$ and $k=3$, which have some fractional coordinates. We used the software \texttt{lrslib} to generate them, and one can check by hand that they are extreme points of the polytope.
    
    To extend these extreme points to $\polytope$ with $n\geq 4$, it is sufficient to set $x_i(\sigs)= 0$ for all $i\geq 4$. To extend them to $k \geq 3$ it is possible to set $\allocs(\sigs) := \allocs(\sigs')$ where $\sigs' = (\min(s_i, 3))_{i\in [n]}$. 
\end{proof}

Finally, when $n\geq3$ and $k\geq 3$, notice that $\polytope$ might be integral for some $\sigmas$, for example when all value and cost functions are constant, in which case $\sigma_i(\osigs) = \emptyset$ and $\polytope$ is just defined by the probability constraints $\sum_i x_i(\sigs)= 1$ for all $\sigs\in[k]^n$.

\begin{figure}[h!]
    \centering
\begin{tikzpicture}[scale=1.4,thick,z=-.25cm]
\draw[->] (0,0,0) -- (3.25,0,0);
\draw[->] (0,0,0) -- (0,3.25,0);
\draw[->] (0,0,0) -- (0,0,3.5);
\begin{scope}[line width=3pt,opacity=.5,blue,line cap=round]
\begin{scope}[loosely dotted]
\draw (0,0,2) -- (1,0,2);
\draw (0,2,0) -- (1,2,0);   
\draw (0,2,1) -- (1,2,1);
\draw (1,0,0) -- (2,0,0);
\draw (1,0,1) -- (2,0,1);
\draw (1,0,2) -- (2,0,2);
\draw (1,1,0) -- (2,1,0);
\draw (1,1,1) -- (2,1,1);
\draw (1,1,2) -- (2,1,2);
\draw (1,2,0) -- (2,2,0);
\draw (1,2,1) -- (2,2,1);
\draw (1,2,2) -- (2,2,2);
\draw (2,0,0) -- (3,0,0);
\draw (2,0,1) -- (3,0,1);
\draw (2,0,2) -- (3,0,2);
\draw (2,1,0) -- (3,1,0);
\draw (2,1,1) -- (3,1,1);
\draw (2,1,2) -- (3,1,2);
\draw (2,2,0) -- (3,2,0);
\draw (2,2,1) -- (3,2,1);
\draw (2,2,2) -- (3,2,2);
\draw (0,0,0) -- (0,1,0);
\draw (0,0,1) -- (0,1,1);
\draw (0,1,0) -- (0,2,0);
\draw (0,1,1) -- (0,2,1);
\draw (0,1,2) -- (0,2,2);
\draw (0,2,0) -- (0,3,0);
\draw (0,2,1) -- (0,3,1);
\draw (0,2,2) -- (0,3,2);
\draw (1,0,0) -- (1,1,0);
\draw (1,0,1) -- (1,1,1);
\draw (1,0,2) -- (1,1,2);
\draw (1,1,0) -- (1,2,0);
\draw (1,1,1) -- (1,2,1);
\draw (1,1,2) -- (1,2,2);
\draw (1,2,0) -- (1,3,0);
\draw (1,2,1) -- (1,3,1);
\draw (1,2,2) -- (1,3,2);
\draw (0,0,0) -- (0,0,1);
\draw (0,0,1) -- (0,0,2);
\draw (0,0,2) -- (0,0,3);
\draw (0,1,0) -- (0,1,1);
\draw (0,1,1) -- (0,1,2);
\draw (0,1,2) -- (0,1,3);
\draw (0,2,2) -- (0,2,3);
\draw (2,0,0) -- (2,0,1);
\draw (2,0,1) -- (2,0,2);
\draw (2,0,2) -- (2,0,3);
\draw (2,1,0) -- (2,1,1);
\draw (2,1,1) -- (2,1,2);
\draw (2,1,2) -- (2,1,3);
\draw (2,2,0) -- (2,2,1);
\draw (2,2,1) -- (2,2,2);
\draw (2,2,2) -- (2,2,3);
\end{scope}
\draw[loosely dotted] (2,3,0) -- (2.5,3,0);
\node[anchor=west] at (2.6,3,0) {1/2};
\end{scope}
\foreach \x in {0,1,2}
\foreach \y in {0,1,2}
\foreach \z in {0,1,2}
\fill (\x,\y,\z) circle (2pt);
\node at (3.5,0,0) {$s_1$};
\node at (0,3.5,0) {$s_2$};
\node at (0,0,4) {$s_3$};
\end{tikzpicture}
\hfill
\begin{tikzpicture}[scale=1.4,thick,z=-.25cm]
\draw[->] (0,0,0) -- (3.25,0,0);
\draw[->] (0,0,0) -- (0,3.25,0);
\draw[->] (0,0,0) -- (0,0,3.5);
\begin{scope}[line width=3pt,opacity=.5,blue,line cap=round]
\begin{scope}[loosely dashed]
\draw (0,1,0) -- (1,1,0);
\draw (0,1,1) -- (1,1,1);
\draw (0,2,0) -- (1,2,0);
\draw (0,2,2) -- (1,2,2);
\draw (1,0,1) -- (2,0,1);
\draw (1,1,1) -- (2,1,1);
\draw (1,2,0) -- (2,2,0);
\draw (1,2,2) -- (2,2,2);
\draw (2,2,0) -- (3,2,0);
\draw (2,2,1) -- (3,2,1);
\draw (2,2,2) -- (3,2,2);
\draw (0,1,2) -- (0,2,2);
\draw (0,2,2) -- (0,3,2);
\draw (1,0,0) -- (1,1,0);
\draw (1,0,2) -- (1,1,2);
\draw (1,1,0) -- (1,2,0);
\draw (1,1,2) -- (1,2,2);
\draw (1,2,0) -- (1,3,0);
\draw (1,2,2) -- (1,3,2);
\draw (0,0,0) -- (0,0,1);
\draw (0,0,1) -- (0,0,2);
\draw (0,0,2) -- (0,0,3);
\draw (0,2,1) -- (0,2,2);
\draw (0,2,2) -- (0,2,3);
\draw (1,2,0) -- (1,2,1);
\draw (1,2,1) -- (1,2,2);
\draw (1,2,2) -- (1,2,3);
\draw (2,0,0) -- (2,0,1);
\draw (2,0,1) -- (2,0,2);
\draw (2,0,2) -- (2,0,3);
\draw (2,1,0) -- (2,1,1);
\draw (2,1,1) -- (2,1,2);
\draw (2,1,2) -- (2,1,3);
\end{scope}
\draw (0,0,2) -- (1,0,2);
\draw (0,1,2) -- (1,1,2);
\draw (1,0,0) -- (2,0,0);
\draw (1,0,2) -- (2,0,2);
\draw (1,1,0) -- (2,1,0);
\draw (1,1,2) -- (2,1,2);
\draw (2,0,0) -- (3,0,0);
\draw (2,0,1) -- (3,0,1);
\draw (2,0,2) -- (3,0,2);
\draw (2,1,0) -- (3,1,0);
\draw (2,1,1) -- (3,1,1);
\draw (2,1,2) -- (3,1,2);
\draw (0,0,0) -- (0,1,0);
\draw (0,0,1) -- (0,1,1);
\draw (0,1,0) -- (0,2,0);
\draw (0,1,1) -- (0,2,1);
\draw (0,2,0) -- (0,3,0);
\draw (0,2,1) -- (0,3,1);
\draw (1,0,1) -- (1,1,1);
\draw (1,1,1) -- (1,2,1);
\draw (1,2,1) -- (1,3,1);
\draw (2,2,0) -- (2,2,1);
\draw (2,2,1) -- (2,2,2);
\draw (2,2,2) -- (2,2,3);
\draw[loosely dashed] (2,3,0) -- (2.5,3,0);
\draw (2,2.75,0) -- (2.5,2.75,0);
\node[anchor=west] at (2.6,3,0) {1/3};
\node[anchor=west] at (2.6,2.75,0) {2/3};
\end{scope}
\foreach \x in {0,1,2}
\foreach \y in {0,1,2}
\foreach \z in {0,1,2}
\fill (\x,\y,\z) circle (2pt);
\node at (3.5,0,0) {$s_1$};
\node at (0,3.5,0) {$s_2$};
\node at (0,0,4) {$s_3$};
\end{tikzpicture}
\caption{Plots of two randomized allocation rules $\allocs$, with $n=3$ and $k=3$. For all $i\in[n]$ and $\sigs\in[k]^n$, the variable $x_i(\sigs)$ is represented by the edge $(s_i,\osigs)$ --- $(s_i+1,\osigs)$, and its value is specified by the type of line and the legend. Observe that both allocations satisfy the monotonicity condition of \Cref{lem:truthful} with monotone value functions: for all $\sigs\in[k]^n$ and for all $i\in[n]$ such that $s_i < k$, we have $x_i(s_i, \osigs) \leq x_i(s_i+1, \osigs)$.
}
    \label{fig:fractional}
\end{figure}
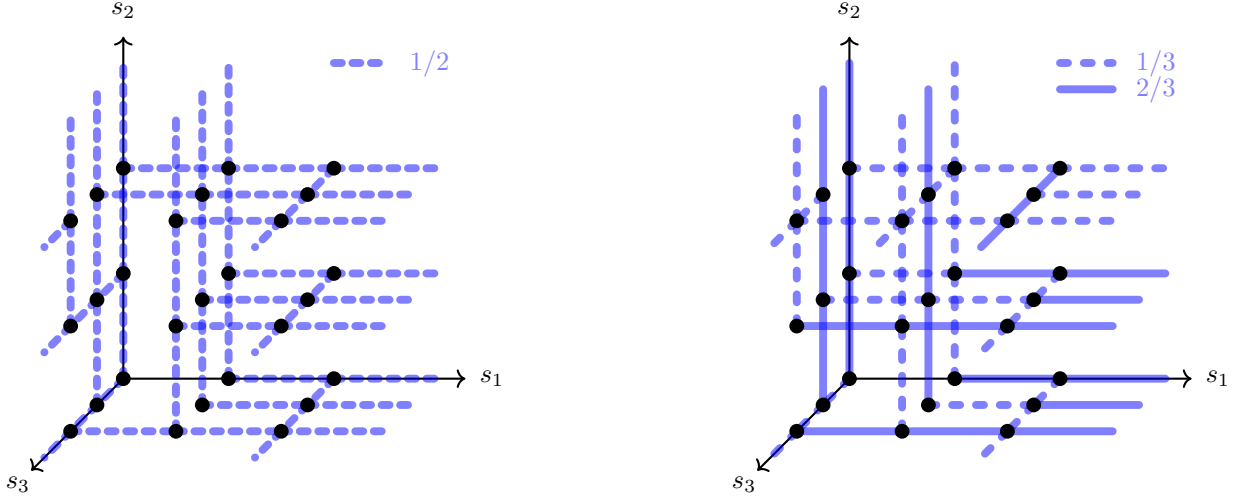

%%%%%%%%%%%%%%%%%%%%%%%%%%%%%%%%%%%%%%%%%%%%%%%%%%%%%%%%%%%%%%%
\subsubsection{Refined Analysis for Two Agents}
\label{sssec:two-agents}

From the analysis in \Cref{sssec:LP,sssec:integrality}, 
we have seen that when $n=2$ we can use linear programming to compute the optimal ratios $\Rvopt$, $\Rcopt$ with \Cref{thm:lp}, and $\Rdoptv$ with \Cref{cor:poly_deter_n2k2}. However, this approach does not directly provide a precise bound on the time complexity. 
In this section, we present a more refined analysis of computing $\Rdoptv$ in this special case by reducing the problem to $2$-SAT, which yields a nearly tight bound on its time complexity.
We will next introduce a directed acyclic graph (DAG) and the concept of conflicting pairs, which enable a precise characterization of $\Rvopt$, $\Rcopt$ and $\Rdoptv$. Leveraging this concept, we develop a dynamic programming  (DP) method to compute the optimal ratios, and a greedy algorithm to obtain the corresponding mechanisms, which leads to a solution of \PbV, \PbC and \PbD in quasi-linear time. 

In the two agents setting, there is a trivial but important property : we could formulate the optimization problems of computing $\Rvopt$, $\Rcopt$, and $\Rdoptv$ only by $x_1$, which simplifies the analysis. This property directly follows from the constraints that $x_2(\sigs)=1-x_1(\sigs)$ for all signal profiles $\sigs$. For this reason, non-decreasing constraints of $x_2$ with respect to $v_2$ can be replaced by non-increasing constraints of $x_1$. We will introduce directed edges between signal profiles, which represent the monotonicity constraints of $x_1$. Due to this property, we will refer to both $x_1$ and $\allocs$ as allocation unless specified otherwise.

\paragraph{Directed graph (DG) preprocessing procedure.} Observe that although $\sigmas$ could fully describe the monotonicity constraints for a given input $\vals$, it may contain redundancy. 
We therefore introduce the following directed graph (DG) preprocessing procedure, which, given input $\vals$, contains the following two steps.
\begin{itemize}
    \item Sort values and partition signals. For a given $i$ and $\osigs$, we sort the values $\{v_i(s_i,\osigs)\}_{s_i\in[k]}$ in non-decreasing order and merge all signal profiles with the same value into blocks $$B_i^{1}(\osigs), B_i^{2}(\osigs), \dots, B_i^{b_i(\osigs)}(\osigs),$$
    where $b_i(\osigs)$ is the number of different values in $\{v_i(s_i,\osigs)\}_{s_i\in[k]}$.
    \item Add dummy vertices and construct a directed graph. We view signal profiles as vertices, introduce dummy vertices to reduce the number of directed edges, and build a directed graph. More precisely, we construct $DG(\vals)$ as follows:
\begin{align*}
         V(\vals) :=&\; [k]^2 
       \cup 
       \bigl\{ a_i^j(\osigs) \,\bigm|\, i\in[2], \osigs \in [k],\, j \in [b_i(\osigs)-1] \bigr\}, \\
       E(\vals) :=&\;
        \bigl\{ (\sigs, a_1^j(s_2)) \,\bigm|\, \sigs \in [k]^2\text{ and }j\in[b_1(s_2)-1]\text{ such that } \sigs \in B_1^j(s_2) \bigr\},\\
        \cup&\;
        \bigl\{ (a_1^j(s_2), \sigs) \,\bigm|\, \sigs\in [k]^2\text{ and }j\in[b_1(s_2)-1]\text{ such that }\sigs \in B_1^{j+1}(s_2) \bigr\},\\
        \cup&\;
        \bigl\{ (\sigs, a_2^j(s_1)) \,\bigm|\, \sigs \in [k]^2\text{ and }j\in[b_2(s_1)-1]\text{ such that }\sigs \in B_2^{j+1}(s_1) \bigr\},\\
        \cup&\;
        \bigl\{ (a_2^j(s_1), \sigs) \,\bigm|\, \sigs \in [k]^2\text{ and }j\in[b_2(s_1)-1]\text{ such that }\sigs \in B_2^{j}(s_1) \bigr\}.
    \end{align*}
\end{itemize}
We observe the following lemma:
\begin{lemma}
    \label{lm:dgpre}
    If $n=2$, one can conduct the DG procedure for an input $\vals$ in time $O(k^2\log k)$, and the resulting directed graph $DG(\vals)$ contains $O(k^2)$ vertices and $O(k^2)$ edges. 
\end{lemma}
\begin{proof}
    The time complexity of the DG procedure is mainly due to the sorting step, which can be done in time $O(k^2\log k)$ by well-known algorithms such as Heapsort or Merge Sort; see any standard algorithms textbook (e.g., \cite{CLRS}). %\MR{TODO: add cite here, the sorting algo}. 
    The number of vertices follows from the fact that for any $i$ and $\osigs$, we have $b_i(\osigs)\le k$. Note that each edge contains a signal profile as one end, and each signal profile appears in at most $4$ edges, which implies that the number of edges is $O(k^2)$.
\end{proof}

We claim further that the truthful constraints can be fully expressed by $DG(\vals)$: $x_1\cfpo\le x_1\cfpt$ if and only if there is a path from $\cfpo$ to $\cfpt$, which directly follows from the fact that the absence of dummy vertices does not add or delete any connectivity between two signal profiles compared to the natural directed graph without adding dummy vertices (which we do not present). We denote the relation that there is a path from $\cfpo$ to $\cfpt$, where $\cfpo\ne\cfpt$, by $\cfpo\prec\cfpt$. Now we are ready to present how to use this procedure to adapt the formulation of \Cref{lem:satisfiability} to the special two agents setting.

%%%%%%%%%%%%%%%%%%%%%%%%%%%%%%%%%%%%%%%%%%%%%%%%%%%%%%%%%%%%%%%
\paragraph{Reduction to $2$-SAT.}

To compute $\Rd^*(\vals)$, we saw in \Cref{sssec:ILP} that one can run a binary search, and check if there exists a deterministic allocation rule which achieves the target approximation ratio. Because in the deterministic case the variables $x_i(\sigs)$ are in $\{0,1\}$, one can encode the probability and monotonicity constraints as a boolean formula, as we have seen in \Cref{lem:satisfiability}. In general, checking the satisfiability of a formula is NP-Hard. However, when $n=2$ the resulting formula is simple and can be solved in linear time~\cite{AspvallPT79}.

\begin{proposition}
    \label{prop:twoagents}
    If $n=2$, one can compute in $O(k^2 \log k)$ the optimal deterministic ratio $\Rdoptv$.
\end{proposition}
\begin{proof}
Given an input $\vals$, we conduct the DG preprocessing procedure, which results in a directed graph $DG(\vals)$ with dummy vertices. We define the following constraints for a given parameter $\alpha\geq 1$, variables $\tilde{x}_1 \in \{0,1\}^{|V(\vals)|}$ which represents the allocation on all $v\in V(\vals)$:
    \begin{align*}
        \forall(v^{\prime},v)\in E(\vals),\qquad
        &\neg \tilde{x}_1(v^{\prime})\vee \tilde{x}_1(v),\\
        \forall v\in[k]^2,\text{ s.t. }\rho_1(v) < 1/\alpha,\qquad& \neg \tilde{x}_1(v),\\
        \forall v\in[k]^2,\text{ s.t. }\rho_2(v) < 1/\alpha,\qquad&  \tilde{x}_1(v).
    \end{align*}
    
    We claim that the existence of a mechanism with approximation ratio at most $\alpha$ is equivalent with the satisfiability of the above formula. Indeed, if we have a feasible solution $\tilde{x}_1$ for the formula, then restrict $\tilde{x}_1$ to $[k]^2$ induces a feasible allocation. If, on the other hand, there exists a feasible allocation $x_1$, it's easy to check that augmenting $x_1$ by setting $x_1(a_1^j(s_2))$ (resp. $x_1(a_2^j(s_1))$) as $\max_{\sigs\in B_1^j(s_2)}x_1(\sigs)$ (resp. $\max_{\sigs\in B_2^{j+1}(s_1)}x_1(\sigs)$) results in a feasible solution for the formula.
    
    Observe that all clauses have size at most $2$ in the above formula, which is a special case of boolean satisfiability, named 2-SAT, for which a satisfying assignment can computed in linear time, for example by transforming each clause into an implication between two literals, and computing the strongly-connected-components of the resulting directed graph~\cite{AspvallPT79}. We further note that we have $O(k^2)$ clauses, which follows directly from \Cref{lm:dgpre}.

    By \Cref{lm:dgpre}, the time complexity of conducting DG procedure is $O(k^2\log k)$. Then we run a binary search on $\Rdoptv$, which performs $O(\log k)$ queries solving a 2-SAT instance in time $O(k^2)$. To conclude, reduction to 2-SAT gives us an algorithm in time $O(k^2\log k)$.
\end{proof}

We now turn to another algorithm \Cref{alg:0 as possible}, which also computes $\Rdoptv$ in $O(k^2\log k)$. Although this algorithm does not improve the time complexity for the computing $\Rdoptv$, we nevertheless present it for several reasons. First, its randomized variant can address the computation of $\Rvopt$ and $\Rcopt$, which cannot be solved by 2-SAT. Second, with monotonicity value or cost, the algorithm achieves a better time complexity bound compared to \Cref{prop:twoagents}. 
Finally, the key concept underlying this algorithm, the conflicting pair, provides an exact characterizations of $\Rvopt$, $\Rcopt$, and $\Rdoptv$. which is of independent interest. 

\paragraph{Directed acyclic graph (DAG) preprocessing procedure} To facilitate our analysis, we introduce the following directed acyclic graph (DAG) preprocessing procedure, building on the DG procedure, and taking $\vals$ as input.
\begin{itemize}
    \item Run the DG procedure. We obtain a directed graph $DG(\vals)$ satisfying \Cref{lm:dgpre}.
    \item Identify all strongly connected components (SCC) in the graph $DG(\vals)$ and contract each component into a single vertex, and let 
    $m(\vals)$ be the total number of components. The edges are contracted accordingly, connecting two components whenever there exists an edge between their corresponding vertices in $DG(\vals)$. The resulting graph is the condensation of $DG(\vals)$, which is a directed acyclic graph (DAG). We refer to it as $DAG(\vals)$. With a slight abuse of notation, we denote the vertex set (resp. edge set) still by $V(\vals)$ (resp. $E(\vals)$).
    \item Perform a topological sort on $DAG(\vals)$, which yields a linear order of its vertices and naturally induces a hierarchical structure. We relabel the  vertices of $DAG(\vals)$, which correspond to strongly connected components of $DG(\vals)$, and denote them by $C^{\mathrm{sc}}_1, C^{\mathrm{sc}}_2, \ldots, C^{\mathrm{sc}}_{m(\vals)}$
\end{itemize}

We present the properties of the DAG procedure in the following lemma.
\begin{lemma}
    \label{lm:dagpre}
    If $n=2$, one can conduct the DAG procedure for an input $\vals$ in time $O(k^2\log k)$, and the resulting directed graph $DAG(\vals)$ contains $O(k^2)$ vertices and $O(k^2)$ edges. 
\end{lemma}
\begin{proof}
    We identify all strongly connected components (SCCs) using Tarjan's linear-time algorithm~\cite{Tarjan1972}, and then perform a topological sort on the resulting condensation graph~\cite[Section~22.4]{CLRS}, which in total take $O(N)$ time.
    Thus, the time complexity of DAG procedure is same as the one of the DG procedure, which is $O(k^2\log k)$. The results for sizes of vertices and edges are trivial, as the graph $DAG(\vals)$ is the condensation of $DG(\vals)$.
\end{proof}

Recall that we denote by $\cfpo\prec\cfpt$ the relation that there is a path from $\cfpo$ to $\cfpt$ in the graph $DG(\vals)$, with $\cfpo\ne\cfpt$. We observe that if $(s_1,s_2)\prec (s_1',s_2')$, 
then either $(s_1,s_2)$ and $(s_1,s_2)$ are in the same component, or the component of $(s_1,s_2)$ appears before that of $(s_1',s_2')$ in the topological order.
Note that all signal profiles within a component $C^{sc}_{j}$ must have a same allocation, which we denote by $x_1(C^{sc}_{j})$.
And we further define a notation $\prec$ over components: $C^{sc}_{j}\,\prec \,C^{sc}_{j^{\prime}}$, if there exists an edge from $C^{sc}_{j}$ to $C^{sc}_{j^{\prime}}$ in $DAG(\vals)$.

Now we are ready to present the conflicting pair characterizations of $\Rvopt$, $\Rcopt$, and $\Rdoptv$ and the algorithms that are based on them.
\paragraph{Conflicting Pairs.}

We first state the our main result in the following theorem.
\thmtwoagents*

We introduce the conflict function to define conflict pair.
\begin{align*}
    f(u,v) :=
\begin{cases}
    1 & \text{if } u=v = 1, \\
    \frac{uv-1}{u+v-2} & \text{otherwise}.
\end{cases}
\end{align*}
One can check that the following inequalities hold:
\begin{equation}
\label{eq:ineq-conflict}
\forall u,v \in(0,1],\qquad 
\frac{1}{f(u,v)} \leq f\left(\frac{1}{u}, \frac{1}{v}\right) \leq \min\left(\frac{1}{u}, \frac{1}{v}\right),
\end{equation}
where the first inequality is direct after the change of variable $(x,y) = (uv, u+v)$, and the second follows from the fact that $f$ is non-decreasing in each coordinate, with $\lim_{u\rightarrow+\infty} f(u,v) = v$.
Now we are ready to present the key notion, conflict pair, in the following:
\begin{definition}
\label{def:cp}
Given $\alpha \geq 1$ and a pair of signal profiles $\cfpo\prec \cfpt$, we say that:
\begin{itemize}
    \item $\cfpo$ and $\cfpt$ form an $\alpha$-value conflict pair if 
$$
    \frac{1}{f(\rho_2\cfpo, \rho_1\cfpt)}>\alpha,
$$
    \item $\cfpo$ and $\cfpt$ form an $\alpha$-cost conflict pair if 
$$
    f\left(\frac{1}{\rho_2\cfpo}, \frac{1}{\rho_1\cfpt}\right)>\alpha,
$$
    \item $\cfpo$ and $\cfpt$ form an $\alpha$-deterministic conflict pair if 
$$
    \min\left(\frac{1}{\rho_2\cfpo}, \frac{1}{\rho_1\cfpt}\right)>\alpha,
$$
\end{itemize}
In particular, using~\Cref{eq:ineq-conflict}, under the same performance ratios $\rhos$, if $\cfpo\prec \cfpt$ is an $\alpha$-value conflict pair then it is an $\alpha$-cost conflict pair; and if it is an $\alpha$-cost conflict pair then it is an $\alpha$-deterministic conflict pair.
\end{definition}

\begin{figure}[h!]
    \centering
    \begin{tikzpicture}[scale=2,>=stealth]
        \small
        \begin{scope}[black!20!white]
        \fill plot[smooth cycle] coordinates {(0.6,1.6)(3.4,1.6)(3.4,3.4)(1.6,3.4)(1.6,2.4)(0.6,2.4)};
        \fill (1,3) circle (.4cm);
        \fill (1,1) circle (.4cm);
        \fill (2,1) circle (.4cm);
        \fill (3,1) circle (.4cm);
        \end{scope}
        \node[circle,draw,fill=red!50!white,minimum size=1.2cm] at (2,1) {$\scriptstyle\!\!(s_1, s_2)\!\!$};
        \node[circle,draw,fill=red!50!white,minimum size=1.2cm] at (1,3) {$\scriptstyle\!\!(s_1', s_2')\!\!$};
        \foreach \i in {1,2,3}
        \foreach \j in {1,2,3}
        \node[draw,minimum size=1.2cm,circle] (\i\j) at (\i,\j) {};
        \begin{scope}[->]
        \draw (12) to [bend right=35] (32);
        \draw (32) -- (33);
        \draw (33) -- (23);
        \draw (23) -- (22);
        \draw[very thick,red] (22) -- (12);
        \draw (22) -- (32);
        \draw (31) to [bend right=35] (33);
        \draw[very thick,red] (21) -- (22);
        \draw (23) -- (13);
        \draw (33) to [bend right=35] (13);
        \draw (11) to [bend left=35] (13);
        \draw (11) -- (12);
        \draw[very thick,red] (12) -- (13);
        \end{scope}
    \end{tikzpicture}
    \hfill
    \begin{minipage}[b]{6.8cm}
    Given $\cfpo$ and $\cfpt$ such that
    \begin{itemize}
    \item $\rho_2(s_1, s_2) = 0.4$,
    \item $\rho_1(s_1', s_2') = 0.5$,
    \end{itemize}
    we have that
    \begin{itemize}
        \item $1/f(\rho_2(s_1, s_2), \rho_1(s_1', s_2')) = 1.375$,
        \item $f(1/\rho_2(s_1, s_2), 1/\rho_1(s_1', s_2')) = 1.6$,
        \item $\min(1/\rho_2(s_1, s_2), 1/\rho_1(s_1', s_2')) = 2$,
    \end{itemize}
    which are respectively lower bounds on $\Rv^*(\rhos)$, $\Rc^*(\rhos)$ and $\Rd^*(\rhos)$.
    \end{minipage}
    \caption{Signal profiles with $n=2$ agents. Given $\cfpo\in[k]^2$, the set of $\cfpt\in[k]^2$ such that $\cfpo\prec\cfpt$ are all signals profiles which are reachable by using one or several arcs.}
    \label{fig:twoagents}
\end{figure}

The notion of conflict pair is illustrated in \Cref{fig:twoagents}. Intuitively, an $\alpha$ conflict pair provides a lower bound of $\alpha$ on the corresponding approximation ratio. Our main result, is that the combination of these lower bounds are tight. When applied to the special case of monotone value (or cost) with respect to signals, our characterization generalizes the \emph{$\alpha$-single crossing} condition from \cite{EdenFFG18}, which was originally defined only for monotone value functions.

\begin{definition}
    \label{alpha_SC}
    \textrm{(adapted from \cite{EdenFFG18})} Given $\alpha\geq 1$, a value setting is said to be $\alpha$-single crossing if for all $i\in[n]$, $\sigs\in[k]^n$, and $\tilde{s_i}\in[k]$ such that $\tilde{s_i}\ge s_i$, we have
    $$
    \forall j\neq i,\qquad \alpha\cdot(v_i(\tilde{s_i},\sigs_{-i})-v_i({s_i},\sigs_{-i}))\ge  v_j(\tilde{s_i},\sigs_{-i})-v_j({s_i},\sigs_{-i}),
    $$
    and a cost setting is said to be $\alpha$-single crossing if for all $i\in[n]$, $\sigs\in[k]^n$, and $\tilde{s_i}\in[k]$ such that $\tilde{s_i}\ge s_i$, we have
    $$
    \forall j\neq i,\qquad\alpha \cdot (c_i(\tilde{s_i},\sigs_{-i})-c_i({s_i},\sigs_{-i}))\le c_j(\tilde{s_i},\sigs_{-i})-c_j({s_i},\sigs_{-i}).
    $$
\end{definition}
We show that the absence of conflict pairs strictly generalizes the $\alpha$-single crossing condition. 
\begin{proposition}\label{prop:single-crossing}
    When $n =2$, for a given $\alpha\ge 1$ if a value or a cost setting is monotone and $\alpha$-single crossing, then there is no $\alpha$-deterministic conflict pair. The converse does not hold when $\alpha>1$.
\end{proposition}
\begin{proof}
    For simplicity, we only provide the proof for the value setting. Note that for monotone value setting, the definition of the partial order $\prec$ becomes: $\cfpo\prec\cfpt$ if and only if $s_1\ge s_1^{\prime}$, $s_2\le s_2^{\prime}$ and $\cfpo\neq \cfpt$.
    Assume for contradiction that there exists an $\alpha$-deterministic conflict pair $\cfpo$ and $\cfpt$, where $\cfpo\prec\cfpt$, which implies that $\max(\rho_2\cfpo,\rho_1\cfpt)<1/\alpha$. Thus, we have $$
    \alpha \cdot v_2\cfpo<v_1\cfpo
    \quad\text{and}\quad
    \alpha\cdot v_1\cfpt<v_2\cfpt.
    $$
    While by $\alpha$-single crossing, we have
    \begin{align*}
    \alpha\cdot (v_1\cfpt-v_1(s_1,s_2^{\prime}))&\ge v_2\cfpt-v_2(s_1,s_2^{\prime})
    \\
    \alpha\cdot(v_2\cfpo-v_2(s_1,s_2^{\prime}))&\ge v_1\cfpo-v_1(s_1,s_2^{\prime}).
    \end{align*}
    From the above inequalities, we obtain that $\alpha\cdot v_1(s_1,s_2^{\prime})< v_2(s_1,s_2^{\prime})$ and $\alpha\cdot v_2(s_1,s_2^{\prime})< v_1(s_1,s_2^{\prime})$, which lead to a contradiction. The inclusion is strict, as illustrated by the following example, where $n= k=2$, which is not $\alpha$-single crossing as $s_2$ has a lot of influence on $v_1$ when $s_1 = 1$, but does not have any $\alpha$-deterministic conflict pair.
    \begin{center}
    \begin{tabular}{|c|c|c|c|c|}
        \cline{1-2}\cline{4-5}
        $v_1(1,1) = 1/\alpha$ & $v_1(1,2) = 1$ &
        &
        $v_2(1,1) = 1/\alpha^2$ & $v_2(1,2) = 1/\alpha^2$
        \\\cline{1-2}\cline{4-5}
        $v_1(2,1) = 1/\alpha$ & $v_1(2,2) = 1$ &
        &
        $v_2(2,1) = 1/\alpha$ & $v_2(2,2) = 1$
        \\\cline{1-2}\cline{4-5}
    \end{tabular}
    \end{center}
\end{proof}
\noindent Now we are ready to state our results.
\begin{lemma}
    \label{lm:alpha_ncf}
    When $n=2$, for all $\alpha\ge 1$, we have that there exists an $\alpha$-approximate in the deterministic (resp. cost or value) setting if and only if there is no $\alpha$-deterministic (resp. $\alpha$-cost or $\alpha$-value) conflict pairs. 
\end{lemma}
This directly yields a proof of \Cref{thm:twoagents}.
\thmtwoagents*
\begin{proof}Using \Cref{lm:alpha_ncf}, the optimal deterministic and randomized approximation ratios have the following explicit expressions:
\begin{align*}
    R^*_v(\vals) = \max_{\cfpo\prec\cfpt}
    &\;\frac{1}{f(\rho_2\cfpo,\rho_1\cfpt)},\\
    R^*_c(\costs) = \max_{\cfpo\prec\cfpt}
    &\;{f\left(\frac{1}{\rho_2\cfpo},\frac{1}{\rho_1\cfpt}\right)},\\
R^*_d(\vals) = \max_{\cfpo\prec\cfpt}
    &\;\min\left(\frac{1}{\rho_2\cfpo},\frac{1}{\rho_1\cfpt}\right).
\end{align*}
We compute the ratio component-wise according to the topological order $C^{\mathrm{sc}}_1, C^{\mathrm{sc}}_2, \ldots, C^{\mathrm{sc}}_{m(\vals)}$. This can be done because if $\cfpo\prec\cfpt$, where $\cfpo\in C_j^{sc}$ and $\cfpt\in C_{j'}^{sc}$, then  $j\leq j'$.
To compute the approximation ratios efficiently, we notice that the functions $f$ and $\min$ are non-decreasing in each coordinate. We denote the minimal $\rho_i$ value of a component $C_{j}^{sc}$ by
$$
\rho_i(C_j^{sc})=\min_{\cfpo\in C_j^{sc}}\rho_i\cfpo,\qquad\text{where }i\in\{1,2\}.
$$
Therefore, for every $C_{j'}^{sc}$ we just need to compute $\rho_1(C_{j'}^{sc})$, $\rho_2(C_j^{sc})$, and the smallest $\rho_2(C_{j'}^{sc})$ over all $C_{j}^{sc}$ such that $C_{j}^{sc} \prec C_{j'}^{sc}$.
Using dynamic programming, one can compute for each $C_{j'}^{sc}$ the quantity 
$$
DP(C_{j'}^{sc}) = \min_{C_j^{sc}\prec C_{j'}^{sc}} \rho_2(C_j^{sc}) = \min\{\rho_2(C_{j'}^{sc}), \min_{(C_j^{sc},C_{j'}^{sc})\in E(\vals)} DP(C_j^{sc})\}.
$$
We observe that the overall time complexity of the DP procedure mainly arises from two types of minimization steps: selecting the minimum within each component, and selecting the minimum across different components. The former can be bounded by the number of vertices in $DG(\vals)$, and the latter can be bounded by the number of edges in $DAG(\vals)$. By \Cref{lm:dgpre} and \Cref{lm:dagpre}, we therefore conclude that the overall time complexity of the dynamic program is $O(N)$.
Once we computed the approximation ratio, the proof of \Cref{lm:alpha_ncf} is constructive and provides mechanisms achieving these ratios, by computing for each $C_{j'}^{sc}$ the largest $x_1(C_j^{sc})$ over all $C_j^{sc}$ such that $(C_j^{sc}, C_{j'}^{sc}) \in E(\vals)$.
\end{proof}
As a sanity check, when $n=2$, notice that \Cref{cor:ratios}, which compares the ratios $\Rvopt$, $\Rcopt$ and $\Rdoptv$, directly follows from the expressions above and the inequalities in \Cref{eq:ineq-conflict}.
\Cref{thm:twoagents} directly implies the following corollary.
\begin{corollary}
    \label{cor:monocfp} When $n=2$, one can compute $\Rvopt$, $\Rcopt$ and $\Rdoptv$ in $O(N)$ if the value (or cost) functions are monotone.
\end{corollary}
\begin{proof}
    We observe that the time complexity $O(N\log N)$ comes only from the value sorting step of DG procedure, and all other steps for computing the optimal ratios take $O(N)$ time. With monotonicity, we can omit the value sorting step.
\end{proof}
We now define notations which will be useful in the algorithms and in the proof of \Cref{lm:alpha_ncf}. To avoid repetition, we provide the proof only for the value setting. Throughout this section, all statements refer to the value setting unless otherwise specified.

The core of the proof relies on the problem's structure being amenable to a \emph{greedy} approach. We will provide a more precise intuition after introducing the linear programming (LP) formulation.
Notice that when $n=2$, an allocation $\allocs$ could be fully described by $x_1$, due to the constraint that $x_1(s_1,s_2)+x_2(s_1,s_2)=1$ for all $(s_1,s_2)\in [k]^2$.
The optimal randomized approximation ratio could therefore be computed by taking the inverse of the optimal value of the following LP:
\begin{subequations}\label{LP_2_agent}
\begin{align}
\text{maximize} \quad & \beta \notag \\ 
\text{such that} \quad 
& 0 \leq x_1(s_1, s_2) \leq 1, 
  && \forall s_1, s_2 \in [k], 
  \label{LP_2_agent:bound} \\
& x_1(s_1, s_2) \cdot \rho_1(s_1, s_2) + (1 - x_1(s_1, s_2)) \cdot \rho_2(s_1, s_2) \geq \beta, 
  && \forall s_1, s_2 \in [k], 
  \label{LP_2_agent:approx} \\
& x_1(s_1, s_2) \leq x_1(s_1', s_2'), 
  && \forall (s_1, s_2) \prec (s_1', s_2'). 
  \label{LP_2_agent:monotonicity}
\end{align}
\end{subequations}
For the deterministic case, we simply add the following integrality constraints to the LP:
\begin{equation}
x_1(s_1, s_2) \in \{0, 1\}, \quad\forall s_1, s_2 \in [k],
\label{eq:integrality_constraint}
\end{equation}
which gives us an integer linear program (ILP).

Notice that the first two types of constraints, constraint (\ref{LP_2_agent:bound}) and (\ref{LP_2_agent:approx}), involve only a single signal profile and can be referred to as \emph{profile-wise constraints}. In contrast, the third type of constraint, constraint \ref{LP_2_agent:monotonicity}, relates two signal profiles and can be referred to as \emph{cross-profile constraints}. Also note that within cross-profile constraints, there are constraints between different layers or in the same strongly connected components. We further refer to the former type of cross-profile constraints as \emph{cross-component constraints} and the latter type as \emph{component-wise constraints}.

The intuition of the proof comes from the fact that the above LP/ILP (\ref{LP_2_agent}) could be tightly solved in the topological order of $DAG(\vals)$. To keep the exposition concise, we illustrate the approach using the LP version. Consider the following procedure: we assign values for all $C^{sc}_{1}\in L_1$, which satisfy the corresponding profile-wise constraints for signal profiles in each component and component-wise constraints, then we select values for  $C^{sc}_{2}$, which satisfy profile-wise, component-wise, and cross-component constraints between $C_1^{sc}$ and $C_2^{sc}$, given the determined value in $C_1^{sc}$, and so on. We note that if this procedure could continue until the assignment of value in $C_{m(\vals)}^{sc}$, then it provides us with a feasible solution for the LP (\ref{LP_2_agent}). 

\begin{algorithm}[h!]
\caption{Zero as Possible (ZaP), when there is no $\alpha$-deterministic conflict pairs, given input $\vals$.}\label{alg:0 as possible}
\begin{algorithmic}[1]
\For{$j' = 1$ to $m(\vals)$}
        \If{$\exists\,(s_1,s_2)\in C^{sc}_{j'}$ such that $\rho_2(s_1,s_2)<\frac{1}{\alpha}$, or\\\qquad $\exists\,C_j^{sc}$, where$\; x_1(C_j^{sc})= 1$ and
        $(C_j^{sc}, C_{j'}^{sc})\in E(\vals)$}
        \State $x_1(C^{sc}_{j'})\gets 1$
        \Else
        \State $x_1(C^{sc}_{j'})\gets 0$
        \EndIf
        \If{$\exists\,(s_1,s_2)\in C^{sc}_{j'}$ such that $x_1(s_1,s_2)\rho_1(s_1,s_2)+(1-x_1(s_1,s_2))\rho_2(s_1,s_2)<\frac{1}{\alpha}$}
        \State \textbf{return} ERROR\label{line: break}
        \EndIf
\EndFor
\State \textbf{return} $\allocs$ 
\end{algorithmic}
\end{algorithm}
\begin{algorithm}[h!]
\caption{Small as Possible (SaP$_v$), when there is no $\alpha$-value conflict pairs, given input $\vals$.}\label{alg:small as possible}
\begin{algorithmic}[1]
\For{$j' = 1$ to $m(\vals)$}
        \State $I_v(C^{sc}_{j'}) \gets
\left\{x\in[0,1]\;\middle|\;
\begin{array}{ll}
x\cdot \rho_1\cfpo+(1-x)\cdot \rho_2\cfpo\ge 1/\alpha, \,\forall (s_1,s_2)\in C^{sc}_{j'},\\
x\ge x_1(C_j^{sc}),\,\,\forall C_j^{sc}\text{ s.t. }(C_j^{sc},C_{j'}^{sc})\in E(\vals).\\
\end{array}\quad\right\}$
        \If{$I_v(C^{sc}_{j'})=\emptyset$}
        \State \textbf{return} ERROR \label{line: error}
        \Else
        \State $x_1(C^{sc}_{j'})\gets \min I_v(C^{sc}_{j'})$
        \EndIf
\EndFor
\State \textbf{return} $\allocs$ 
\end{algorithmic}
\end{algorithm}
\begin{algorithm}[h!]
\caption{Small as Possible (SaP$_c$), when there is no $\alpha$-cost conflict pairs, given input $\costs$.}\label{alg:small as possible cost}
\begin{algorithmic}[1]
\For{$j' = 1$ to $m(\costs)$}
        \State $I_c(C^{sc}_{j'}) \gets
\left\{x\in[0,1]\;\middle|\;
\begin{array}{ll}
x/ \rho_1\cfpo+(1-x)/ \rho_2\cfpo\ge \alpha, \,\forall (s_1,s_2)\in C^{sc}_{j'},\\
x\ge x_1(C_j^{sc}),\,\,\forall C_j^{sc}\text{ s.t. }(C_j^{sc},C_{j'}^{sc})\in E(\vals).\\
\end{array}\quad\right\}$
        \If{$I_c(C^{sc}_{j'})=\emptyset$}
        \State \textbf{return} ERROR \label{line: error2}
        \Else
        \State $x_1(C^{sc}_{j'})\gets \min I_c(C^{sc}_{j'})$
        \EndIf
\EndFor
\State \textbf{return} $\allocs$ 
\end{algorithmic}
\end{algorithm}

Our algorithms are parametrized by a parameter $\alpha\geq 1$, such that no $\alpha$ conflict pair exists. One can pre-compute the smallest $\alpha$ such that no such conflict pair exist through the explicit expressions.

\begin{proof}[Proof of \Cref{lm:alpha_ncf}]
We start by the \emph{deterministic} case. To prove the \emph{if} direction, we assume that there is no $\alpha$-deterministic conflict pair, that is, for all %$\cfpo$ and $\cfpt$ such that 
$\cfpo\prec\cfpt$ we have
$$
    \max(\rho_2\cfpo, \rho_1\cfpt) \geq 1/\alpha.
$$

We run \Cref{alg:0 as possible} which, if successful, builds an $\alpha$-approximate deterministic mechanism, and we show it never enters line \ref{line: break}:
\begin{itemize}
    \item If  $x_1(C_{j'}^{sc})=0$, then we cannot trigger line \ref{line: break}, as we have that $\rho_2(s_1,s_2)\ge\frac{1}{\alpha}$ for all $\cfpo\in C_{j'}^{sc}$;
    \item If  $x_1(C_{j'}^{sc})=1$ and $\exists\,\cfpo\in C_{j'}^{sc}\,\,\text{such that}\,\,\rho_2(s_1,s_2)<\frac{1}{\alpha}$, then we cannot trigger line \ref{line: break}, as for all $\cfpt\in C_{j'}^{sc}$ we have $\rho_1\cfpt=1$;
    \item If  $x_1(C_{j'}^{sc})=1$ and $\rho_2(s_1,s_2)\geq\frac{1}{\alpha}$ for all $\cfpo\in C_{j'}^{sc}$, then there exists $C_{j}^{sc}$, such that $x_1(C_{j}^{sc})= 1$ and $(C_j^{sc},C_{j'}^{sc})\in E(\vals)$. In this subcase, there must exist a profile $(s_1^{\prime\prime},s_2^{\prime\prime})\in C_{j''}^{sc}$ such that $C_{j''}^{sc}\prec C_{j'}^{sc}$ and $\rho_2(s_1^{\prime\prime},s_2^{\prime\prime})<\frac{1}{\alpha}$. As otherwise, all $C_{j}^{sc}$ such that $(C_j^{sc},C_{j'}^{sc})\in E(\vals)$ would have $x_1(C_j^{sc})=0$. Recall that with the absence of $\alpha$-deterministic conflict pair, we have that $\max(\rho_2\cfpo, \rho_1(s_1^{\prime\prime},s_2^{\prime\prime}))\ge \frac{1}{\alpha}$, which implies that $\rho_2(s_1,s_2)\ge\frac{1}{\alpha}$, for all $\cfpt\in C_{j'}^{sc}$. Thus, line \ref{line: break} is also not triggered in this subcase.
\end{itemize}
We turn to prove the \emph{only if} direction. By contradiction, assume that there exists an $\alpha$-deterministic conflict pair $(s_1,s_2)$ and $(s_1^{\prime},s_2^{\prime})$, where $(s_1,s_2)\prec$ $(s_1^{\prime},s_2^{\prime})$ and
$$\max(\rho_2(s_1,s_2),\rho_1(s_1^{\prime},s_2^{\prime}))<1/\alpha.$$ 
To achieve $\alpha$-approximation, we have to set $x_1(s_1,s_2)=1$ and $x_1\cfpt=0$. However, monotonicity constraints require that $ x_1(s_1,s_2)\leq x_1(s_1^{\prime},s_2^{\prime})$, which yields a contradiction.

We now address the \emph{value} case. First, consider the \emph{if} direction. Assume that there is no $\alpha$-value conflict pair, that is
for all $(s_1,s_2)\prec (s_1',s_2')$ we have
$$
f(\rho_2(s_1,s_2),\rho_1(s_1',s_2')) \geq \frac{1}{\alpha}.
$$

We run \Cref{alg:small as possible} which, if successful, builds a randomized mechanism with a value approximation ratio of at most $\alpha$, and we prove that it never reaches line \ref{line: error}.
The key to the process succeeding is that, in each component, the intervals defined by profile-wise constraints and cross-component constraints have a non-empty intersection. Note that component-wise constraints satisfy automatically if the resulting interval is non-empty, as we allocate the same value for all signal profiles in the component, which is exactly what component-wise constraints ask for.

For that matter, we introduce the following quantities:
\begin{equation}
lp(s_1,s_2) := 
\begin{cases}
\frac{1/\alpha-\rho_1(s_1,s_2)}{1-\rho_1(s_1,s_2)}, & \text{if } \rho_1(s_1,s_2) \neq 1, \\[6pt]
-\infty, & \text{if } \rho_1(s_1,s_2) = 1,
\end{cases}
\label{eq:l_definition}
\end{equation}
and
\begin{equation}
rp(s_1,s_2) :=
\begin{cases}
\frac{1-1/\alpha}{1-\rho_2(s_1,s_2)}, & \text{if } \rho_2(s_1,s_2) \neq 1, \\[6pt]
+\infty, & \text{if } \rho_2(s_1,s_2) = 1,
\end{cases}
\label{eq:r_definition}
\end{equation}
which represent the interval endpoints defined by constraint (\ref{LP_2_agent:approx}). And we further define:
$$lp(C_{j'}^{sc}) = \max_{\cfpt\in C_{j'}^{sc}}lp\cfpt\quad\text{and}\quad rp( C_{j'}^{sc})=\min_{\cfpt\in C_{j'}^{sc}}rp\cfpt.$$
It directly follows from the definitions that $lp\cfpo\le1$ ,$rp\cfpo\ge0$, $lp( C_{j'}^{sc})\le1$ and $rp(C_{j'}^{sc})\ge 0$ hold for all $\cfpt$ and $C_{j'}^{sc}$.
More precisely, we have
\begin{equation}
\label{eq:Iinter}
I_v(C_{j'}^{sc})=[0,1]\cap [lp(C_{j'}^{sc}),rp(C_{j'}^{sc})] \cap\bigcap_{C_j^{sc}\prec C_{j'}}[x_1(C_{j}^{sc}), +\infty),
\end{equation}
the last part appears when $\ell\ge2$.
Next, we claim that for all $(s_1,s_2)$ and $(s_1^{\prime},s_2^{\prime})$ such that $(s_1,s_2)\prec (s_1^{\prime},s_2^{\prime})$, the inequality
\begin{equation}
\label{rgel}
lp\cfpo\leq rp\cfpt,
\end{equation}
follows directly from the definitions of $lp$ and $rp$, and that $
f(\rho_2(s_1,s_2),\rho_1(s_1^{\prime},s_2^{\prime})) \geq 1/\alpha
$. This inequality further implies that 
\begin{equation}
    \label{ineq:l&rcomp}
    lp(C_{j}^{sc})\leq rp(C_{j'}^{sc}),
\end{equation}
where $C_{j}^{sc}\prec C_{j'}^{sc}$. One can easily check that the  inequalities $lp\cfpt\le rp\cfpt$ and $lp(C_{j'}^{sc})\le rp(C_{j'}^{sc})$ hold for the same reasons.
Finally, we show by induction on $j'\in[m(\vals)]$ that the algorithm does not fail during the first $j'$ steps, and that
\begin{equation}
x_1(C_{j'}^{sc}) = \max \left(\{0\}\cup\{lp(C_{j'}^{sc})\}\cup\{lp(C_{j}^{sc})\}_{C_{j}^{sc}\prec C_{j'}^{sc}}\right)\leq \min\left(1, rp(C_{j'}^{sc})\right).
\end{equation}
\newcommand{\SCCo}{C^{sc}_{1,j}}
\newcommand{\SCCplus}{C^{sc}_{\ell+1,j^{\prime}}}
The induction hypothesis directly holds for $j' = 1$, as we have $x_1(C_{1}^{sc}) = \min I_v(C^{sc}_{1})$ with
$$
I_v(C^{sc}_{1}) = [0,1] \cap [lp(C^{sc}_{1}), rp(C^{sc}_{1})].
$$
Next, assuming the induction hypothesis holds at $j'<m(\vals)$, then for $C_{j'+1}^{sc}$, we use \Cref{eq:Iinter} and the transitivity of $\prec$ to show that
$$
I_v(C_{j'+1}^{sc}) =[0,1]\cap[lp(C_{j'+1}^{sc}), rp(C_{j'+1}^{sc})]\cap\bigcap_{C_{j}^{sc}\prec C_{j'+1}^{sc}}[lp(C_{j}^{sc}), +\infty).
$$
Then, using \Cref{ineq:l&rcomp} and other inequalities derived above we obtain that $I_v(C_{j'}^{sc})\neq\emptyset$, which proves that the algorithm does not fail at step $j'+1$ and that $x_1(C_{j'+1}^{sc})$ satisfies the induction hypothesis. We end our induction, having shown that the algorithm never fails, and having provided an efficient computable definition of $x_1(C_{j}^{sc})$.

Now we turn to the proof of the \emph{only if} direction. By contradiction, we assume that there exists an $\alpha$-value conflict pair, which indicates that there exists $\cfpo$ and $\cfpt$, where $\cfpo\prec\cfpt$ and
$$f(\rho_2\cfpo,\rho_1\cfpt)< \frac{1}{\alpha}.$$
As $\alpha\ge1,$ we have that $f(\rho_2\cfpo,\rho_1\cfpt)<1$, which implies $\rho_2\cfpo\neq1$ and $\rho_1\cfpt\neq 1$. To have $\alpha$-approximate randomized mechanism, it requires that 
$$
x_1\cfpt\le rp\cfpt = \frac{1-1/\alpha}{1-\rho_1\cfpt}
\quad\text{and}\quad
x_1\cfpo\ge lp\cfpo = \frac{1/\alpha-\rho_2\cfpo}{1-\rho_2\cfpo}.
$$
It follows from monotonicity constraints that
$$
\frac{1/\alpha-\rho_2\cfpo}{1-\rho_2\cfpo} \le x_1\cfpo \le x_1\cfpt\le \frac{1-1/\alpha}{1-\rho_1\cfpt},
$$
which contradicts the assumption that $\cfpo$ and $\cfpt$ is an $\alpha$-value conflict pair.
\end{proof}

%%%%%%%%%%%%%%%%%%%%%%%%%%%%%%%%%%%%%%%%%%%%%%%%%%%%%%%%%%%%%%%
\subsubsection{Refined Analysis for Binary Signals}
\label{sssec:two-signals}

Recall that we can run a binary search on the optimal deterministic ratio, reducing \PbD to the query variant $\PbD_\gamma$. Building on the insights developed in \Cref{sssec:hypergraph}, we now turn to another special case, $k=2$, for which the formulation of our decision problem as a perfect matching in a hypergraph can be solved efficiently. 

\thmtwosignals*

\begin{proof}
To compute $\Rdoptv$ and the corresponding deterministic mechanism, we will run the binary search of \Cref{lem:binary-search}. To answer $\PbD\gamma$ queries, we turn our attention to \Cref{lem:hypergraph}. Recall there exists a solution if and only if the following hypergraph has a perfect matching:
$$
V := [k]^n\qquad\text{and}\qquad
E := \bigcup_{i\in[n]}
\left\{
e\in E_i\;|\; \forall \sigs\in e, \rho_i(\sigs) \geq 1/\gamma
\right\}$$
where for all $i\in [n]$ we have
$$
    E_i = \left\{\{(s_i, \osigs)\;|\; s_i \in S\}\;\middle|\;\begin{array}{l}
    S\subseteq [k]\text{ and }\osigs\in [k]^{n-1}\text{ such that}\\
    \forall (s_i, s_i')\in \sigma_i(\osigs), \; s_i\in S \Rightarrow s_i'\in S\end{array}\right\}.
$$
First, we observe that each hyperedge $e\in E$ has size at most $2$, and we split $E$ into the set $E'$ of hyperedges of size $1$ and the set $E''$ of hyperedges of size $2$. We have that $E$ contains at most $n\cdot3\cdot2^{n-1}$ elements, thus $E'$ and $E''$ can be constructed in $O(n2^n)$ time. Intuitively, $(V,E'')$ is a standard undirected graph, and $E'\subseteq V$ can be thought as a set of nodes which are allowed to be left unmatched. More formally, a solution is a matching in $(V,E'')$ such that each vertex $\sigs\in V$ is either covered by one edge $e\in E''$, or can be left unmatched if $\sigs\in E'$.

Now, observe that each edge $e\in E''$ connects two signal profiles which only differ on the $i$-th coordinate. If we partition vertices $\sigs\in [2]^n$ using the parity of $\sum_i s_i$, we obtain that $(V,E'')$ is a bipartite graph. 

Interestingly, we can show that for every edge $\{(1, \osigs), (2, \osigs)\} \in E''$, we either have $(1, \osigs)\in E'$ or $(2, \osigs)\in E'$. Indeed, if $(1,\osigs)\notin E'$ then $(1,2)\in \sigma_i(\osigs)$ and if $(2,\osigs)\notin E'$ then $(2,1)\in \sigma_i(\osigs)$, which would contradict the fact that $\sigma_i(\osigs)$ is a strict order.

To simplify the structure even further, we define $E''' = E''\setminus\{\{\sigs,\sigs'\}\;|\;\sigs,\sigs' \in E'\}$ by removing the edges which connect two vertices of $E'$ that could be left alone, which does not change the existence of a solution. In the graph $G_\gamma = (V, E''')$ each edge covers exactly one edge from $V\setminus E'$. 

Therefore, there exists a solution if and only if $G$ has a matching of size $|V|-|E'|$. We will compute a maximum cardinality matching in the bipartite graph $G$, which has $2^n$ vertices and $N = O(n2^n)$ edges, which can be done in time $N^{1+o(1)}$ using the quasi-linear algorithm of \cite{ChenKLPGS25}, or in time $O(N^{3/2})$ with the (more standard) Hopcroft–Karp–Karzanov algorithm \cite{HopcroftK73}.
\end{proof}

%%%%%%%%%%%%%%%%%%%%%%%%%%%%%%%%%%%%%%%%%%%%%%%%%%%%%%%%%%%%%%%
%%%%%%%%%%%%%%%%%%%%%%%%%%%%%%%%%%%%%%%%%%%%%%%%%%%%%%%%%%%%%%%
%%%%%%%%%%%%%%%%%%%%%%%%%%%%%%%%%%%%%%%%%%%%%%%%%%%%%%%%%%%%%%%
\section{Hardness and Lower Bounds}
\label{sec:negative-results}

In this section, we prove \Cref{thm:nphard,thm:query}, which respectively give lower bounds on the time complexity and query complexity of the optimization problems we consider. In both proofs, we build instances with increasing value functions, so that our hardness results apply to the setting considered in previous works \cite{RoughgardenT16,EdenFFG18,EdenFFGK19,AmerT21,LuSZ22}.

%%%%%%%%%%%%%%%%%%%%%%%%%%%%%%%%%%%%%%%%%%%%%%%%%%%%%%%%%%%%%%%
%%%%%%%%%%%%%%%%%%%%%%%%%%%%%%%%%%%%%%%%%%%%%%%%%%%%%%%%%%%%%%%
\subsection{Build valuation functions from performance ratios}
\label{ssec:build_valuation}

We have seen in \Cref{lm_find_value_for_performance_ratios} that one can find increasing value (or decreasing cost) functions inducing the given performance ratios. Thus, to build hard instances for monotone value (or cost) settings in this section, we instead focus on performance ratios. In this subsection, we further prove that one can always build monotone value (or cost) functions with submodularity over signals (SOS) based on performance ratios that satisfy a simple condition. Therefore, a weaker version of hardness results could be extended to the SOS setting.
 
We recall the definition of submodularity over signals from \cite{EdenFFGK19}: we say that value functions exhibit \emph{submodularity over signals (SOS)} if for all $i$, $j\in[n]$, $s_j\in[k]$, and $\Delta \in [k]$, where $s_j+\Delta\le k$, and for any $\sigs_{-j}, \sigs'_{-j}\in [k]^{n-1}$ such that $\sigs_{-j} \le \sigs'_{-j}$ component-wise, it holds that
\[
v_i(s_j+\Delta, s_{-j}) - v_i(s_j, s_{-j})
\;\ge\;
v_i(s_j+\Delta, s'_{-j}) - v_i(s_j, s'_{-j}).
\]
Similarly, we define submodularity over signals (SOS) of cost functions as
\[
c_i(s_j, s_{-j}) - c_i(s_j+\Delta, s_{-j})
\;\ge\;
c_i(s_j, s'_{-j}) - c_i(s_j+\Delta, s'_{-j}).
\]

Given performance ratios $\rhos$, define $r^{*}(\rhos)=\min\{\rho_i(\sigs)\;|\;i\in [n], \sigs\in [k]^n\}$.
%\lemvaluationgeneral*
%\lemvaluationsos*
\begin{restatable}{lemma}{lemvaluationsos} Given performance ratios $\rhos$ where $r^{*}(\rhos)\ge 1-\frac{1}{(nk)^2+1}$, one can construct increasing value functions (resp. decreasing cost functions) which are SOS and induce $\rhos$.
    \label{lm_can_find_sub_valuefunc}
\end{restatable}
\begin{proof}
    \label{pf_can_find_sub_valuefunc}
    Observe that if value functions $\{v_i\}_{i\in[n]}$ are increasing and SOS, then the cost functions obtained by taking the reciprocal of the value functions are decreasing, SOS, and induce the same performance ratios with the value functions. Therefore, it suffices to prove the lemma in the value setting. We define $a_{\ell} = \ell(2nk-\ell)$ and $v_i(\sigs)=\rho_i(\sigs)(a_{\ell}+1)$, where $\ell=||\sigs||_1$. It is easy to check that for all $i$ and $\sigs$ the performance ratio induced by $v_i(\sigs)$ is $\rho_i(\sigs)$. As $1-\frac{1}{(nk)^2+1}\le\rho_i(\sigs)\le 1$, we have $a_{\ell}\le v_i(\sigs)\le a_\ell +1$. We further note that $a_{\ell^{\prime}}\ge a_{\ell}+1$ if $\ell^{\prime}>\ell$. Thus, $v_i(\sigs)\ge v_i(\sigs^\prime)$ if $||\sigs||_1<||\sigs^\prime||_1$, which implies the monotonicity of the value functions. To show SOS of the value functions , we notice that $a_{\ell+\Delta}-a_{\ell}-1\le v_i(s_j+\Delta,\sigs_{-j})-v_i(s_j,\sigs_{-j})\le a_{\ell+\Delta}-a_{\ell}+1$ and $a_{\ell+\Delta}-a_\ell\ge a_{\ell^\prime+\Delta}-a_{\ell^\prime}+2$ when $\ell<\ell^\prime$, as $a_{\ell+\Delta}-a_\ell=\Delta(2nk-2\ell-t)$ and $\Delta\ge1$, which show that the value functions are SOS.
\end{proof}

%%%%%%%%%%%%%%%%%%%%%%%%%%%%%%%%%%%%%%%%%%%%%%%%%%%%%%%%%%%%%%%
%%%%%%%%%%%%%%%%%%%%%%%%%%%%%%%%%%%%%%%%%%%%%%%%%%%%%%%%%%%%%%%
\subsection{NP-Hardness}
\label{ssec:nphard}

We saw in \Cref{prop:twoagents} that solving \PbD can be reduced to the satisfiability of formulas with clauses of size at most $n$, which is easy when $n=2$ and NP-Hard when $n \geq 3$. In this section we give reverse reduction, proving that in general computing, and even approximating the optimal deterministic ratio $\Rd^*$ is NP-Hard. To have a simpler reduction, we start from the 1-in-3-SAT problem, a structured variant of 3-SAT which is also NP-Hard \cite{Schaefer78}. Importantly, we will need to have $n=4$ agents, as embedding an arbitrary formula within an instance with only three agents is not feasible with our current construction. We leave the complexity of computing $\Rd^*$ when $n=3$ as an intriguing open question.

\begin{definition}[1-in-3-SAT]
In the 1-in-3-SAT problem, we are given a boolean formula $\phi$ in {conjunctive normal form (CNF)}, where each clause consists of exactly three literals (i.e., variables or their negations). The goal is to determine whether there exists a truth assignment to the variables such that \emph{exactly one} literal in each clause is true, and the other two are false.
Formally, let
\[
\phi = C_1 \land C_2 \land \dots \land C_m,
\]
where each clause $C_i$ has the form $(\ell_{i,1}, \ell_{i,2}, \ell_{i,3})$, and each $\ell_{i,j}$ is a literal (either a variable $x$ or its negation $\neg x$). The formula $\phi$ is said to be \textit{1-in-3 satisfiable} if there exists a truth assignment such that, for each clause $C_i$, exactly one of the literals $\ell_{i,1}, \ell_{i,2}, \ell_{i,3}$ evaluates to true.
\end{definition}

\thmnphard*
\begin{proof}
In the gap problem $(1, \beta)$-\PbD, we are asked to distinguish between instances with deterministic ratio $\Rd^*=1$ and $\Rd^*>\beta$. For convenience, we fix a constant $\varepsilon \in (0, 1/\beta)$.

Given $\phi = C_1 \land C_2 \land \dots \land C_m$, we are going to define an instance of $(1,\beta)$-\PbD with $n=4$ agents and $k=O(m)$ signals, which has a polynomial size in $m$. If $\phi$ is 1-in-3 satisfiable, then the deterministic ratio will be equal to $1$, otherwise it will be equal to $1/\varepsilon > \beta$. 

For simplicity, we build an instance with monotone value functions $\vals$, that is, such that $\sigma_i(\osigs) = \{(s_i,s_i')\;|\; 1 \leq s_i < s_i' \leq k\}$ for all $i\in[n]$ and $\osigs$. By \Cref{lm_find_value_for_performance_ratios}, it suffices to build performance ratios $\rhos$. In our construction, we will need to set each performance ratio $\rho_i(\sigs)$ to be either equal to $1$ or $\varepsilon$. %For that matter, we define $v_i(\sigs) = \rho_i(\sigs)/(2\varepsilon)^{\sum_{i} s_i}$, which can be described with $O(m)$ bits. By construction, each value function $v_i$ is increasing in $s_i$, for every possible choice of performance ratio. 

To build some intuition, consider the gadget for a variable $a$ in \Cref{fig:hardness-var}. First, because of the performance ratios, if $\Rd(\allocs, \rhos) = 1$ then we have
$$
x_1(k,1,1,*) =x_4(k,1,1,*) = x_2(1,2,1,*) = x_4(1,2,1,*) = x_3(1,1,2,*) = x_4(1,1,2,*) = 0.
$$
If we decide to set $x_1(1,1,2,*) = 1$, then
\begin{align*}
x_1(1,1,2,*) = 1
&\quad\Rightarrow\quad
x_1(k,1,2,*) = 1
& (\text{by monotonicity})\\
&\quad\Rightarrow\quad
x_3(k,1,2,*) = 0
& (\text{sum of proba is 1}) \\
&\quad\Rightarrow\quad
x_3(k,1,1,*) = 0
& (\text{by monotonicity}) \\
&\quad\Rightarrow\quad
x_2(k,1,1,*) = 1
& (\text{sum of proba is 1}) \\
&\quad\Rightarrow\quad \dots\\
&\quad\Rightarrow\quad
x_1(1,1,2,*) = 1
\end{align*}
Thus, all these are equivalent, and if they hold we say that $a$ is true. Conversely, we say that $a$ is false if $x_1(1,2,1,*) = 1$ holds. Note that we introduced the fourth agent to set $x_4(1,1,1,*) = 1$, as monotonicity prevents agents $1$, $2$ and $3$ to be selected at $(1,1,1,*)$. We introduce one gadget per variable, located in distinct $s_2$ and $s_3$ so that to remove any unwanted interaction between these gadgets. The horizontal line at coordinates $(s_2,s_3+1)$ is called the $a$-line, and the horizontal line at $(s_2+1, s_3)$ is called the $\neg a$-line. When $a$ is true (resp. false) we say that the $a$-line is active (resp. inactive), and that the $\neg a$-line is inactive (resp. active).

Next, we build one gadget per clause $C_i$ using \Cref{fig:hardness-clause}. Each clause is made of three variable gadgets, one for each literal $\ell_{i,j}$, which intersect at a single signal profile $\sigs$, for which we set $\rhos(\sigs) = (1,1,1,\varepsilon)$. Thus, we have to select one winner $j\in\{1,2,3\}$, which corresponds to the literal which is set to true (exactly one literal will be true).

Finally, we connect each literal $\ell$ with the corresponding variable, using the XOR-connector gadget of \Cref{fig:hardness-connect}. When two lines are connected by a XOR-connector gadget, it adds the constraint that exactly one of the two lines is active. Importantly, each literal has two horizontal lines in their gadget, but one of these two lines does not have unique coordinates, and thus cannot be connected. Thus, we connect the horizontal line with unique coordinates to either the $a$ line or the $\neg a$ line, depending on the sign of the literal and on the sign of the horizontal line.

One can check that if there exists a truth assignment such that each clause has exactly one literal true, then one can build a deterministic mechanism $\allocs$ with $\Rd(\allocs, \rhos) = 1$, proving that $\Rd^*(\rhos) = 1$. More precisely, we use the assignment for every gadget, we propagate according to monotonicity, and we set $x_4(\sigs) = 1$ for any other signal profile. Conversely, if $\Rd^*(\rhos) = 1$ then we can build a satisfying assignment for $\phi$. Thus, distinguishing between if $\Rd^*(\rhos) = 1$ and $\Rd^*(\rhos) > \beta$ is NP-Hard.
\end{proof}

Next, as a corollary, we want to prove that the mechanism design problem is also hard. Intuitively, answering $\PbD_\gamma$ queries is related to the search variant of \textsc{Sat}, which is known to be polynomialy equivalent to its decision variant, using Cook reductions.
\begin{corollary}
    \label{cor:QueryPb}
    Assuming $\text{P}\neq\text{NP}$, for every $\gamma\geq 1$ one cannot always answer $\PbD_\gamma$ queries in polynomial time. The $\PbD_\gamma$ problem remains hard even with the stronger promise that there exists a deterministic allocation rule of ratio $1$.
\end{corollary}
\begin{proof}
    We proceed by contrapositive, and assume that we have an algorithm which can answer $\PbD_\gamma$ queries in $f(N) = N^{O(1)}$ time on instances where there exists a deterministic allocation rule of ratio $1$.

    Given an instance $\paras$ of $(1,\beta)$-\PbD with $n=4$ agents and $\beta=\gamma$, we query the algorithm on all $k^n$ signal profiles, each time stopping after $f(N)$ steps if the algorithm has not stopped yet. We check if the resulting allocation rule $\allocs$ is truthful and has a ratio $R(\rhos, \allocs) \leq \beta$. We return true if it is the case, and we return false in any other case of failure. By construction, the procedure we just described can decide $(1,\beta)$-\PbD in polynomial time, which implies that $P=NP$.
\end{proof}

We notice that above hardness results could be extended to SOS setting when $\varepsilon$ is sufficiently large. Formally, we have the following corollaries:

\begin{corollary}
    \label{cor:NPSOS} When $n=4$, the problem $(1,1)$-\PbD is NP-hard, even with monotone SOS value (or cost) functions.
\end{corollary}
\begin{proof}
    Take $\varepsilon\in(1-\frac{1}{(nk)^2+1},1)$, and apply the same proof of \Cref{thm:nphard}. By \Cref{lm_can_find_sub_valuefunc}, there always exist monotone SOS functions inducing the given performance ratios.
\end{proof}
\begin{corollary}
    \label{cor:queryPBSOS}Assuming $\text{P}\neq\text{NP}$, one cannot always answer $\PbD_1$ queries in polynomial time, even on instances where there exists a deterministic allocation rule of ratio $1$ and the (value or cost) functions are monotone and SOS.
\end{corollary}
\begin{proof}
    Combine the proof of \Cref{cor:QueryPb} with the result \Cref{cor:NPSOS}.
\end{proof}

%%%%%%%%%%%%%%%%%%%%%%%%%%%%%%%%%%%%%%%%%%%%%%%%%%%%%%%%%%%%%%%
\subsection{Query Complexity}
\label{ssec:query}

\Cref{thm:lp,thm:twoagents,thm:twosignals} give algorithms to solve \PbV, \PbC, and special cases of \PbD in polynomial time, with respect to the total size of the input $N = nk^k b$. These algorithms can be used to solve the corresponding mechanism design questions $\PbV_\gamma$, $\PbC_\gamma$ and $\PbD_\gamma$, where the output allocation rule can be evaluated at different signal profiles, with the constraint of being consistent across queries. However, one might wonder if we can drop the exponential dependency in $n$ in the time complexity when we are only asked to compute the outcome at one of the $k^n$ signal profiles.

We answer this question using the query complexity in the decision tree model. More formally, we assume that our algorithm can access the input via an oracle which answers value or cost queries. We build a set of instances for which there exists a deterministic allocation rule of ratio $1$, but for which any algorithm cannot correctly compute a valid outcome at a specific signal profile without querying most of the input.

\thmquery*

\begin{proof}
    We will build a set of instances where the input value functions are monotone increasing. Similarly to the proof of \Cref{thm:nphard}, we can first specify the performance ratio $\rhos$ then build the value functions $\vals$ such that $\sigma_i(\osigs) = \{(s_i,s_i')\;|\; 1 \leq s_i < s_i' \leq k\}$ for all $i\in[n]$ and $\osigs$.
    
    We start with the proof of $\PbD_1$, for which the main idea is the following. First we fix a specific signal profile $\bar \sigs$, and we define initial performance ratios $\bar\rhos$. Then we build two sets of instances $P_1$ and $P_2$ such that the following properties hold:
    \begin{itemize}
        \item every $\rhos\in P_1\cup P_2$ differs from $\bar\rhos$ at exactly one performance ratio $\rho_i(\sigs)$,
        \item for all $\rhos\in  P_1\cup P_2$ there exists $\allocs\in \intpolytope$ with ratio $R(\rhos, \allocs) = 1$,
        \item if $\rhos\in P_1$ then for all $\allocs\in\intpolytope$ such that $R(\rhos, \allocs) = 1$ we have $x_1(\bar\sigs) = 1$,
        \item if $\rhos\in P_2$ then for all $\allocs\in\intpolytope$ such that $R(\rhos, \allocs) = 1$ we have $x_2(\bar\sigs) = 1$.
    \end{itemize}
    Then, we either need to query all the performance ratios where an instance from $P_1$ differ from $\bar\rhos$, or all the performance ratios where an instance from $P_2$ differ from $\bar\rhos$, otherwise one cannot compute the outcome $\allocs(\bar\sigs)$.

    Next, our goal is to build $\bar\rhos$ such that for each $\sigs\in[k]^n$ there exists at most two agents $i$ such that $\bar\rho_i(\sigs) = 1$, and all other have $\bar\rho_i(\sigs) = \varepsilon$. This way, our $\PbD$ instance can be understood as a 2-SAT formula with variables $x_i(\sigs)$, where each clause is an implication between two literals. Let $L_1$ and $L_2$ be respectively the sets of literals implied (transitively) by $x_1(\bar\sigs)$ and $x_2(\bar\sigs)$, using the monotonicity constraints of $\sigmas$. We will make sure that $L_1$ and $L_2$ are disjoint, and both have size $\Omega(k^n)$. This way, we define $P_1$ and $P_2$ as the sets of instances which differ from $\bar\rhos$ on the variables associated to the literals of $L_1$ and $L_2$, forcing the corresponding literal to be false.
    
    We can now proceed with the construction of $\bar\sigs$ 
    $$
        \forall i\in[n],\qquad
        \bar s_i = \begin{cases}
        1+\lfloor k/2\rfloor & \text{if }i\in\{1,2\},\\
        1 & \text{if }i\geq 3\text{ and }i\text{ is odd},\\
        k & \text{if }i\geq 3\text{ and }i\text{ is even}.
    \end{cases}
    $$
    To construct $\rhos$, we first deal with $\sigs$ such that $s_i = \bar s_i$ for all $i\geq 3$:
    \begin{itemize}
        \item if $s_1 \geq \bar s_1$ then we set $\rho_2(\sigs) = 1$,
        \item if $s_2 \geq \bar s_2$ then we set $\rho_1(\sigs) = 1$,
        \item if $n \geq 3$ and either $s_1 < \bar s_1$ or $s_2 < \bar s_2$ we set $\rho_3(\sigs) = 1$,
        \item for every other $i$ we set $\rho_i(\sigs) = \varepsilon$.
    \end{itemize}
    Observe that the construction for agents $1$ and $2$ is similar to the one given in \Cref{ssec:ideas-query} when $n=2$. Then, for the remaining $\sigs\in[k]^n$ we define $J(\sigs)$ as the largest $j$ such that $s_j\neq\bar s_j$, and we set $\rho_i(\sigs) = 1$ for all $i$ such that $J(\sigs) \leq i \leq J(\sigs)+1$. With this construction, one can show by induction on $i\geq 3$ that we have:
    \begin{itemize}
        \item for all odd $i\geq 3$, for all $\sigs\in[k]^n$ such that $J(\sigs) = i$,
        \begin{itemize}
            \item if $s_1\geq \bar s_1$ and $s_2 < \bar s_2$ then $x_i(\sigs)\in L_1$
            \item if $s_1< \bar s_1$ and $s_2 \geq \bar s_2$ then $x_i(\sigs)\in L_2$
        \end{itemize}
        \item for all even $i\geq 4$, for all $\sigs\in[k]^n$ such that $J(\sigs) = i$,
        \begin{itemize}
            \item if $s_1\geq \bar s_1$ and $s_2 < \bar s_2$ then $\neg x_i(\sigs)\in L_1$
            \item if $s_1< \bar s_1$ and $s_2 \geq \bar s_2$ then $\neg x_i(\sigs)\in L_2$
        \end{itemize}
    \end{itemize}
    Thus, both sets have size $\Omega(k^n)$. Moreover, one can show that the sets $L_1 \subseteq\{\sigs\in[k]^n\;|\; s_1\geq \bar s_1\text{ and }s_2<\bar s_2\}$ and $L_2\subseteq\{\sigs\in[k]^n\;|\; s_1\geq \bar s_1\text{ and }s_2<\bar s_2\}$ are contained in separate quadrants of the set of signal profiles and therefore are disjoint. This concludes the construction of $\rhos$, and the proof for $\PbD_1$. To finish the proof of the theorem, observe that with the exact same sets of instance:
    \begin{itemize}
        \item if $\rhos\in P_1$ then for all $\allocs\in\polytope$ such that $\Rv(\rhos, \allocs) = 1$ we have $x_1(\bar\sigs) = 1$,
        \item if $\rhos\in P_2$ then for all $\allocs\in\polytope$ such that $\Rv(\rhos, \allocs) = 1$ we have $x_2(\bar\sigs) = 1$.
    \end{itemize}
    The argument in the cost setting is identical. This proves that it requires at least $\Omega(k^n)$ queries to answer some $\PbV_1$, $\PbC_1$ queries.
\end{proof}

By \Cref{lm_can_find_sub_valuefunc}, we extend \Cref{thm:query} to SOS setting as follows:
\begin{corollary}
    For all fixed $n\geq 2$ and $k\geq 2$, it requires at least $\Omega(k^n)$ queries to the input oracle (value or cost) to answer some $\PbV_1$, $\PbC_1$ and $\PbD_1$ queries, even with monotone SOS (value or cost) functions.
\end{corollary}
\begin{proof}
    Take $\varepsilon\in(1-\frac{1}{(nk)^2+1},1)$, and apply the same proof of \Cref{thm:query}. By \Cref{lm_can_find_sub_valuefunc}, there always exist monotone SOS functions inducing the given performance ratios.
\end{proof}

% Bibliography
\clearpage
\printbibliography

\end{document}